\documentclass[conference,10pt]{IEEEtran}

\pdfoutput=1

\makeatletter
\def\ps@headings{%
\def\@oddhead{\mbox{}\scriptsize\rightmark \hfil \thepage}%
\def\@evenhead{\scriptsize\thepage \hfil \leftmark\mbox{}}%
\def\@oddfoot{}%
\def\@evenfoot{}}
\makeatother
\usepackage{amsthm}
\usepackage{setspace} 
\usepackage{graphicx}
\usepackage{epstopdf}
\usepackage{caption,subcaption} 
\usepackage{algorithm}

\usepackage{cite}
\usepackage{cleveref}

\usepackage[noend]{algorithmic} 

\usepackage{xstring}

\newcommand{\algwhole}[1]{\ifthenelse{\equal{#1}{p}}%
                     {Procedure CombinedUpdate}{CombinedUpdate}}

\newcommand{\depUpdate}[1]{\ifthenelse{\equal{#1}{p}}%
                     {Procedure DepthUpdate}{DepthUpdate}}

\newcommand{\greedy}[1]{\ifthenelse{\equal{#1}{p}}%
                     {Procedure GreedyTreeCover}{GreedyTreeCover}}

\newcommand{\single}[1]{\ifthenelse{\equal{#1}{p}}%
                     {Procedure SingleTreeAdjust}{SingleTreeAdjust}}

\newcommand{\mixnode}[1]{\ifthenelse{\equal{#1}{p}}%
                     {Procedure MixedNodeAdjust}{MixedNodeAdjust}}

\newcommand{\internal}[1]{$#1$-internal}
\newcommand{\parent}[1]{$#1$-parent}
\newcommand{\child}[1]{$#1$-child}

\newcommand{\leaf}[1]{$#1$-leaf}
\newcommand{\leafs}[1]{$#1$-leaves}

\newcommand{\mix}[2]{mixed-$#1$-$#2$-node}

\newcommand{\flo}[1]{\lfloor #1 \rfloor}
\newcommand{\cei}[1]{\lceil #1 \rceil}
\newcommand{\tZ}{\tilde{Z}}

\newtheorem{theorem}{Theorem}[section]

\newtheorem{prop}[theorem]{Proposition}
\newtheorem{lemma}[theorem]{Lemma}

\newtheorem{assumption}{Assumption}
\newtheorem{defi}{Definition}

\newcommand{\post}[2]{
\centering \leavevmode
 \includegraphics[width=#2cm]{#1}
 }

\crefname{assumption}{assumption}{assumptions}
\crefname{prop}{proposition}{propositions}
\floatname{algorithm}{Procedure}
\crefname{algorithm}{procedure}{Procedures}

\title{Tree dynamics for peer-to-peer streaming}
\author{
\authorblockN{Ji Zhu, Bruce Hajek}
\authorblockA{
 University of Illinois at Urbana-Champaign\\
 jizhu1@illinois.edu,  b-hajek@illinois.edu
 }
 
}

\begin{document}

\maketitle
\begin{abstract}
This paper presents an asynchronous distributed algorithm to manage multiple trees for peer-to-peer streaming in a flow level
model. It is assumed that videos are cut into
substreams, with or without source coding, to be distributed to all nodes.
The algorithm guarantees that each node receives sufficiently many
substreams within delay logarithmic in the number of peers. The algorithm works by constantly updating the topology so
that each substream is distributed through trees to as many nodes
as possible without interference. Competition among trees for limited
upload capacity  is managed so that both coverage and balance are
achieved. The algorithm is robust in that it efficiently eliminates
cycles and maintains tree structures in a distributed way. The algorithm
favors nodes with higher degree, so it not only works for live streaming
and video on demand, but also in the case a few nodes with large degree
act as servers and other nodes act as clients.

A proof of convergence of the algorithm is given assuming instantaneous
update of depth information, and for the case
of a single tree it is shown that the convergence time is
stochastically tightly bounded by a small constant times
the log of the number of nodes.   These theoretical results are
complemented by simulations showing that the algorithm works well even when
most assumptions for the theoretical tractability do not hold.




\end{abstract}

\section{Introduction}
Peer-to-peer communication, when applied to streaming, is attractive because
it enables the total bandwidth for service to scale with demand \cite{liu2010uusee,zhang2005coolstreaming}.
Scheduling algorithms need to be designed so that upload capacity provided by nodes can be
utilized efficiently to serve the demand, and each node can download streams, providing playback continuity with small delay. The system needs to be
robust enough to tolerate peer churn, link failures, and congestion.

Many designs have been proposed for P2P streaming.   Unstructured topologies with nodes constantly sampling new
targets for new pieces are considered in \cite{zhou2007simple,wu09}.   Data dissemination algorithms based on
mesh topologies are given in \cite{magharei2009prime,kostic2003bullet}. In \cite{kim_srikant_cycle},
random Hamiltonian cycles are constructed and tangled and pieces are broadcasted around the
union of the cycles. Fixed underlying topologies are considered in \cite{tomozei2010flow,nguyen2010chameleon} and flows are scheduled between neighbors. Algorithms to manage multiple distribution trees to disseminate different substreams of a video or
audio are discussed in  \cite{padmanabhan2003resilient,castro2003splitstream,tran2003zigzag,zhang2012overlay},  but those papers do not concentrate on distributed algorithms.

Unstructured streaming systems are simple to manage and scale well,  but playback continuity is sacrificed because constantly
building and removing links requires prohibitive overhead.   Tree structures can provide good
playback continuity with small startup delay, but can be difficult to manage, especially in a distributed way.
In this paper we study how to manage trees for P2P streaming as in \cite{padmanabhan2003resilient,castro2003splitstream,tran2003zigzag}, but with a focus on distributed
algorithms.



Consider a complete underlying network for control information, so arbitrary node to node contact is allowed. To model the bandwidth constraint for control information, we only allow each node to randomly contact a target  from other nodes periodically at a certain constant rate. This setting, which is more suitable for P2P systems, is different from settings in \cite{Brosh_multicast,Gallager_MST,Jacobsen_multicast,shi2002routing,widmer2012rate}, which discuss how to build multicast trees satisfying certain metrics under a fixed underlying topology. Most problems formulated in \cite{Brosh_multicast,Jacobsen_multicast,shi2002routing,widmer2012rate} are shown to be NP hard and approximation algorithms are designed. In this paper we avoid NP completeness  with a homogeneous and complete  underlying topology. 

 The streaming network is built on top of the underlying network, through the cooperation of nodes. Besides the bandwidth constraint on exchanging control information, nodes have heterogeneous upload capacity so each of them has a maximum fan-out degree for streaming.  Nodes have a small buffer to store information about their parents and children in the streaming network. They can exchange messages with their parents and children, at the same time they can also exchange messages with their sampled targets at the sampling times.
As in \cite{padmanabhan2003resilient,castro2003splitstream,tran2003zigzag}, we assume the video is cut into substreams, source coding like multiple description coding (MDC \cite{goyal2001multiple}) is applied to provide redundancy in data. Multiple diverse distribution trees, which constitute the streaming network, are constructed and managed, each for one substream. 

In \cite{padmanabhan2003resilient} it is mentioned that a good tree management algorithm should maintain 1) short trees, i.e., trees have small depths so as to minimize the probability of disruption due to peer transience or congestion; 2) tree diversity, i.e., the set of ancestors of a node in each tree are as disjoint as possible so as to increase the effectiveness of the MDC-based distribution scheme; 3) quick processing of node joins and leaves, and; 4) scalability.  Centralized solutions in \cite{padmanabhan2003resilient,castro2003splitstream,tran2003zigzag} are proposed to achieve those goals.
In this paper, an asynchronous distributed algorithm is designed to manage multiple trees, with the following properties:
\begin{enumerate}
\item Each node can receive enough substreams.
\item Depths of trees are logarithmic in the number of peers.
\item Trees are diverse and balanced.
\item Cycles are eliminated efficiently in a distributed way.
\item Convergence is fast, providing robustness to peer transience.
\item Nodes with higher upload capacity tend to be closer to the roots of the trees.
\item  Heterogeneous upload capacity is supported, even in the case a few nodes
with large degree act as servers and other nodes act as clients.
\item  Convergence is insured even when the ratio of total demand to total upload capacity, $\rho$, is one, and it
converges more quickly as $\rho$ decreases from one.
\end{enumerate}

Analyzing the complexity for message exchange for the algorithm is quite challenging. We show that the convergence time is stochastically tightly bounded by $O(\log N)$ where $N$ is the number of nodes, both by theoretical analysis in the case of a single tree and by simulation. During a sampling interval $O(N)$ messages are exchanged. So we conjecture that with high probability,
only $O(N\log N)$ messages need to be exchanged before the algorithm converges. 


Other related work for P2P message transmission includes \cite{cohen03,zhu_swarm,ji2011,norros2011stability,massoulie2005,menasch10,bruce2011,oguz2012stable,sanghavi2007gossiping}, which focus on modeling the performance of file sharing networks.

The paper is organized as follows. The model is introduced in \Cref{sec_problem}. The main algorithm and the proof of convergence are provided in \Cref{sec_alg}. Bounds on the convergence time are covered in \Cref{sec_converge}.   Simulations are provided
in \Cref{sec.sim} which show that the algorithm works well even when most assumptions imposed in \Cref{sec_alg} for theoretical tractability do not hold.

\section{Problem setup}\label{sec_problem}

Consider a network containing one server and $N$ users (nodes), labeled as $1,2,...N$. One video to be broadcast from the server to all nodes is cut into $M$ substreams.  Each substream is transmitted through a directed broadcast tree. We consider the problem of how to build broadcast trees, so as to avoid interference, achieve coverage and reduce delay.

\begin{figure}
\post{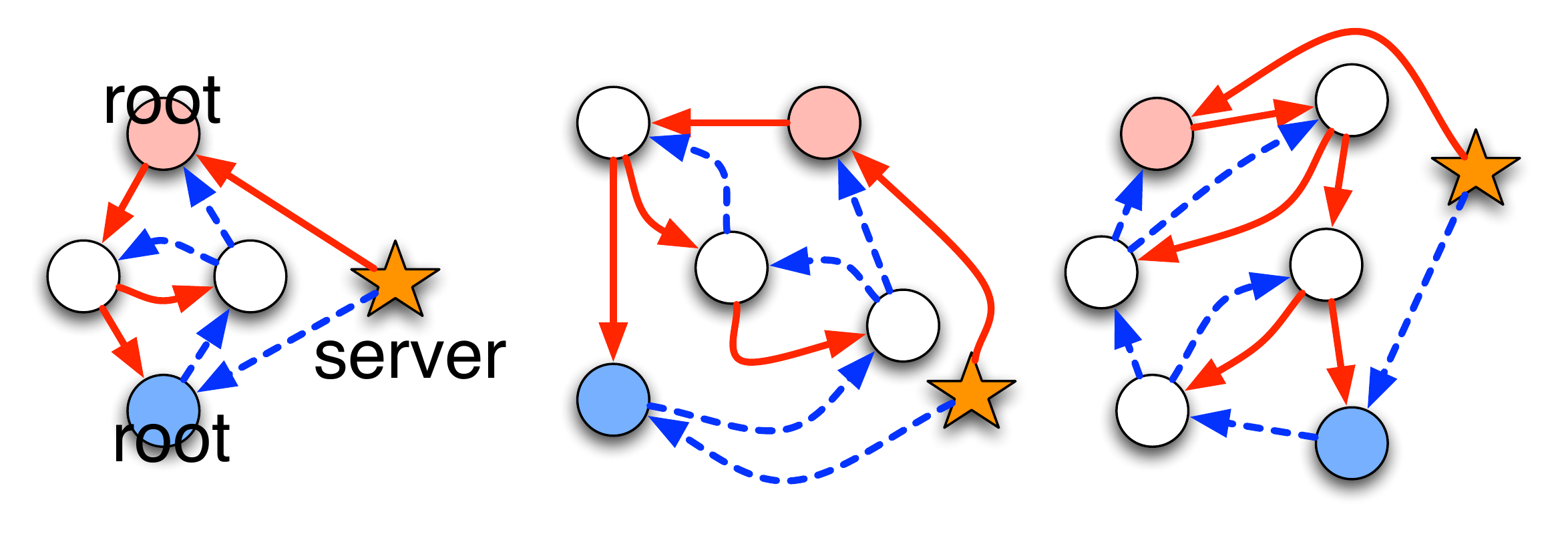}{6.5}
\caption{Spanning trees with $\log N$ depth, $N=4,5,6$.}
\label{fig_pTree}
\end{figure}

In this paper we consider a flow level model. Let $V$ denote the set of $N$ nodes. As illustrated in \Cref{fig_pTree}, assume
there are $M$ root nodes (or roots, $M<< N$) in $V$, each of which always has an incoming link from the server, and always receives a distinct substream via such link. Each root works as an ``agent'' of the server to further distribute its received substream.
Let $\mathcal{R}$  denote the set of roots. For convenience, assume the root nodes are nodes 1 through $M$, and
label each substream by the label of the root receiving the stream from the server.

Assume nodes in $V$ can randomly contact each other and build directed links among themselves. Let $E$ denote the set of all such links. Through each link one and only one substream can be transmitted from the tail to the head of the link. Assume that at the time a link is built, the substream to be transmitted on it is also determined.   Assume each link is colored by the label of the substream transmitted on it. Let $E_i$ denote the set of all links with color $i$ for $i\in\mathcal{R}$. The set of all $E_i$'s is a partition of $E$.
{\em A node $u$ can receive substream $i$ if and only if in graph $(V,E_i)$ there is a directed path from root $i$ to $u$; and the delay of receiving substream $i$ is modeled by the number of hops of the shortest path from $i$ to $u$}. Let $V_i$ denote the set of nodes to which there exists a path from root $i$ in $(V,E_i)$. That is, $V_i$ is the set of nodes which can receive substream $i$.

Assume that to recover the video, a node $u$ must receive at least $K\leq M$ substreams: $|\{i\in\mathcal{R}: u\in V_i\}|\geq K$. 
It is possible that $K<M$, corresponding to the use of source coding.  Assume each node has a constraint on upload capacity
(outdegree): node $u$ can build at most $\bar{d}_u$ outgoing links, whatever their colors are.
\Cref{fig_pTree} contains examples of $(V,E)$ for $M=K=2$ and $N=4,5,6$, where roots have outdegree one and other nodes have outdegree two. 
In \Cref{fig_pTree},  for each root $i$, $(V,E_i)$ is a spanning tree with the minimum depth under out degree constraint, with the tree depth defined as the maximum number of hops over all root-leaf paths.

\begin{figure}
\post{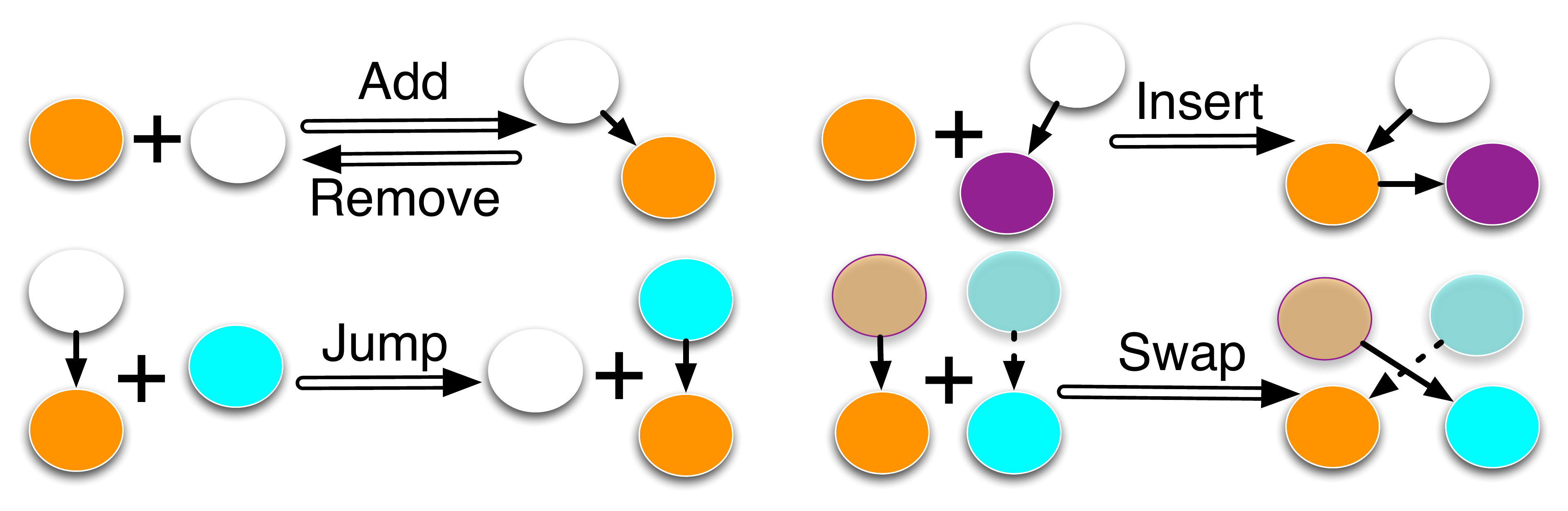}{7}
\caption{All five types of link updates. Link colors in the ``Swap'' transform may or may not be the same.}
\label{fig_linkChange}
\end{figure}

Five types of link updates are considered in this paper, as shown in \Cref{fig_linkChange}. Each link update can be
executed locally because at most four nodes are involved. 
Notice that link updates in \Cref{fig_linkChange} are just combinations of building a link and removing a link.
Any node on the left side of a link update in \Cref{fig_linkChange} can initiate that update, by exchanging messages with other nodes involved. 

Assume each node maintains a Poisson clock which ticks at rate $\mu=1$, independently of Poisson clocks of other nodes. Whenever the clock of a node ticks, the node samples a target from other nodes uniformly at random, and decides whether to execute link updates in \Cref{fig_linkChange} or not.
 {\em In this paper, we assume link updates happen instantaneously, but at most one update can be executed at each sampling. } It makes the problem tractable and is also a reasonable relaxation because one link update consists of building at most two links and removing at most two links. In this paper we normalize the time so $\mu=1$.

In \Cref{fig_pTree}, the spanning tree for each substream has
$\log N$ depth. In the next section, we show our algorithm, under which $E$ is repeatedly updated until no more updates
are possible, insures that all substreams can be broadcasted through distinct trees with $\log N$ depth and each node can 
receive enough substreams for the video.   For convenience, notations are listed: 
\begin{itemize}


\item{$L_i(u):$} the \em{depth} to $i$ of $u$, defined as the minimum number of hops from root $i\in\mathcal{R}$ to $u$ in $(V,E_i)$.
Define $L_i(u) = +\infty$ if no path exists from $i$ to $u$ in $(V,E_i)$. 
\item{$l_i^u:$} the depth to $i$ buffered by $u$, which keeps updating.
\item{$V_i: $} for each $ i\in\mathcal{R}, V_i: = \{u:u\in V, L_i(u)<+\infty\}$.
\item{$i$-link: } an $i$-link refers to a link colored $i$.
\item{\parent{i}: }  $u$ is an \parent{i} of $v$ if $(u,v)\in E_i$.
\item{\child{i}: } $v$ is an \child{i} of $u$ if $(u,v) \in E_i$.
\item{\leaf{i}: } Node $u$ is an \leaf{i} if $d_i(u) = 0$.
\item{\internal{i}: } Node $u$ is  an \internal{i} node if $L_i(u) + 2\leq \max_{w\in V_i} L_i(w)$.
\item{$d_i(u): $} the number of outgoing $i$-links of node $u$. Let $d(u) = \sum_{i\in\mathcal{R}} d_i(u)$ be the total number of outgoing links of $u$.
\item{available: } Node $u$ is \em{available} if $u$ has an incoming link and $d(u) < \bar{d}_u$.
\item{mixed node:} A node is mixed if it has outgoing links with at least two colors.
\item{\mix{i}{j}: } Node $u$ is a \mix{i}{j} if $d_i(u)d_j(u) >0$.
\item Write $(a_{i}) < (b_{i})$ for two sequences to mean $\forall i, a_{i} < b_{i}$. 
\end{itemize}


\section{Algorithm}\label{sec_alg}
Our main algorithm is summarized in \Cref{sec_whole} as \algwhole{p}, which is a topology update  procedure ran after nodes randomly contact each other. In the following we introduce \algwhole{0} by parts. Notice that  by assumption in \Cref{sec_problem} running of \algwhole{0} is instantaneous. 
Assume each node knows its parents, children,  and
colors of its incoming and outgoing links.  Each node buffers
its depths to all roots. 
To begin, we assume at time $0$ $(V,E)$ satisfies \Cref{a.init}.

\begin{assumption}[\bf{Tree Initially}]\label{a.init}
~

At time $0$, $(V,E)$ satisfies the following:
\begin{enumerate}
\item Each node $u$ has at most $K$ incoming links, has at most one incoming $i$-link for each $i\in\mathcal{R}$, and at most $\bar{d}_u$ outgoing links, which implies that,
\item $(V_i,E_i)$ is a directed tree rooted at $i$ for each $i\in\mathcal{R}$.
\end{enumerate}
\end{assumption}
Notice that  1) implies 2) in \Cref{a.init}. And notice that cycles may appear in $(V\setminus V_i, E_i)$ even if $(V_i,E_i)$ is a tree,  as shown in \Cref{f.cycle}. \Cref{a.init} can be easily satisfied by requiring all nodes to remove extra incoming or outgoing links at time $0$. We will show that the properties in \Cref{a.init} are preserved by the update procedures we shall define.
In the following the procedure for depth update is discussed first.

\subsection{Distributed depth update}
\label{sec.dDist}
Each node buffers its depths to each root $i\in\mathcal{R}$. Notice $l_i^u$ denotes the depth to $i$ buffered by node $u$. 

\begin{algorithm}
\renewcommand{\thealgorithm}{}
\caption{\protect\depUpdate{0} (Node $u$, Color $i$)}
\label{a.depUpdate}

If node $u$ is root $i$, $u$ sets $l_i^u$ to $0$; otherwise,
\begin{itemize}
\item if $u$ has no incoming $i$-links, $u$ sets $l_i^u$ to $+\infty$;
\item if $u$ has incoming $i$-links, let $l$ be the minimum depth to $i$ buffered by all \parent{i}s of $u$, $u$ sets $l_i^u$ to $l+1$.
\end{itemize}
\end{algorithm}

We think $\infty+1=\infty$ in  \depUpdate{0}. Root $i$ sets its depth to $i$ as $0$, and any other node updates its depth to $i$ as $l+1$ if it finds that the minimum depth to $i$ buffered by its \parent{i}s is $l$. Each node $u$ periodically runs \depUpdate{0} so as to insure its buffered depths are close approximations to the real depths $L_i(u), i\in\mathcal{R}.$  We will show that each node maintains just one incoming $i$-link over all time, so a node needs to contact just one \parent{i} when \depUpdate{0} is running. 

We do not require nodes update depths much more frequenctly than they sample targets. Assume node $u$ updates depths for three types of events:
\begin{itemize}
\item after $u$ builds a new incoming $i$-link, immediately $u$ runs \depUpdate{0}($u$,$i$). 
\item after $u$ samples a target, $u$ runs \depUpdate{0}($u$,$j$) for all $j\in\mathcal{R}$ immediately. 
\item after $u$ is sampled as a target, $u$ runs \depUpdate{0}($u$,$j$) for all $j\in\mathcal{R}$ immediately.
\end{itemize}
Thus, nodes update depths about twice as fast as their Poisson clocks tick. 

\subsection{Distributed greedy covering}

In this section  a greedy procedure is proposed to  insure that each node has at least $K$ incoming links with distinct colors. Notice that nodes are randomly sampling others. Assume a node $u$ runs \greedy{p} after it samples node  $u_p$ as a target:
\begin{algorithm}
\renewcommand{\thealgorithm}{}
\caption{\protect\greedy{0} (Node $u$, Node $u_p$)}
\label{a.cover}

{\bf Output: } return true if $(V,E)$ changes, return false otherwise.

{\bf Tie broken: } arbitrarily

If $u$ has less than $K$ incoming links and there exists $i$ such that $u$ has no incoming $i$-link but $u_p$ has an incoming $i$-link,
\begin{itemize}
\item{\bf Add:} if $d(u_p)<\bar{d}_{u_p}$, build $(u_p,u)$ in $E_i$, and return true;
\item{\bf Insert:} if $u_p$ has an \child{i}, say $u_c$, remove $(u_p, u_c)$, build $(u_p,u),(u,u_c)$ in $E_i$, and return true.
\end{itemize}
Return false. (See \Cref{f.addinsert})
\end{algorithm}

\begin{figure}
\post{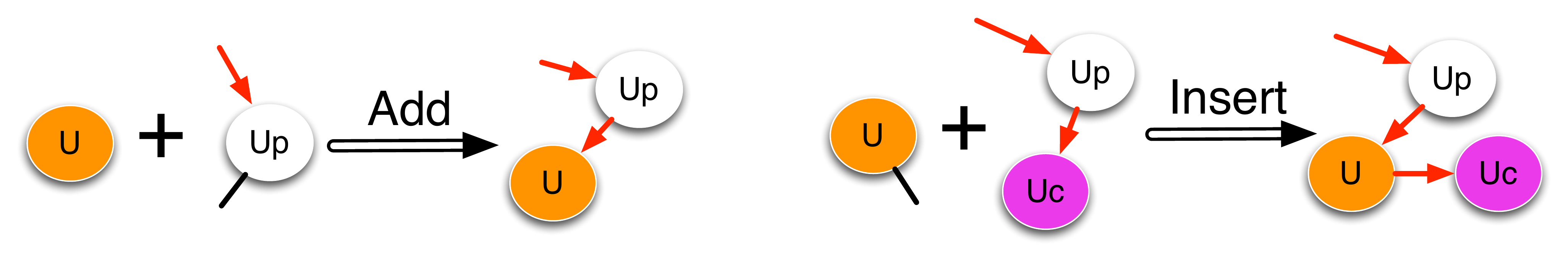}{7.5}
\caption{Greedy Tree Cover}
\label{f.addinsert}
\end{figure}

\greedy{0} does not use depth information. It has several properties: if \Cref{a.init} is valid, as nodes sample targets and run \greedy{0}, 
\begin{enumerate}
\item statements in \Cref{a.init} remain valid;
\item for each $i\in\mathcal{R}$, $|V_i|$ is nondecreasing, but the depth of tree $(V_i,E_i)$ is also nondecreasing.
\item for each $i\in\mathcal{R}$ and each node $u$, both $d_i(u)$ and the number of incoming $i$-links of $u$ are nondecreasing.
\end{enumerate}

In the following we state several additional assumptions under which running \greedy{0} after nodes sample targets can lead all nodes to be covered by $K$ trees. 

First, nodes have to provide enough fan-out degrees to meet the demands of incoming links, so \Cref{a.link} is assumed.
\begin{assumption}[\bf{Minimum Degree}]\label{a.link}
$\sum_{u\in V}\bar{d}_u \geq KN - M$.
\end{assumption}
Second, notice that for $i\in\mathcal{R}$, if at time $0$ root $i$ is not available and root $i$ does not have outgoing $i$-links, it is not possible for any node to ever build incoming $i$-links from $i$. That is, $V_i = \{i\}$ over all time. To avoid that, we assume root $i$ has at least one \child{i} at time $0$:
\begin{assumption}[\bf{Root Child Guarentee}]\label{a.root}
For each $i\in\mathcal{R}$, $\bar{d}_i \geq 1$ and at time $0$ root $i$ has at least one \child{i}. 
\end{assumption}

\begin{figure}
\post{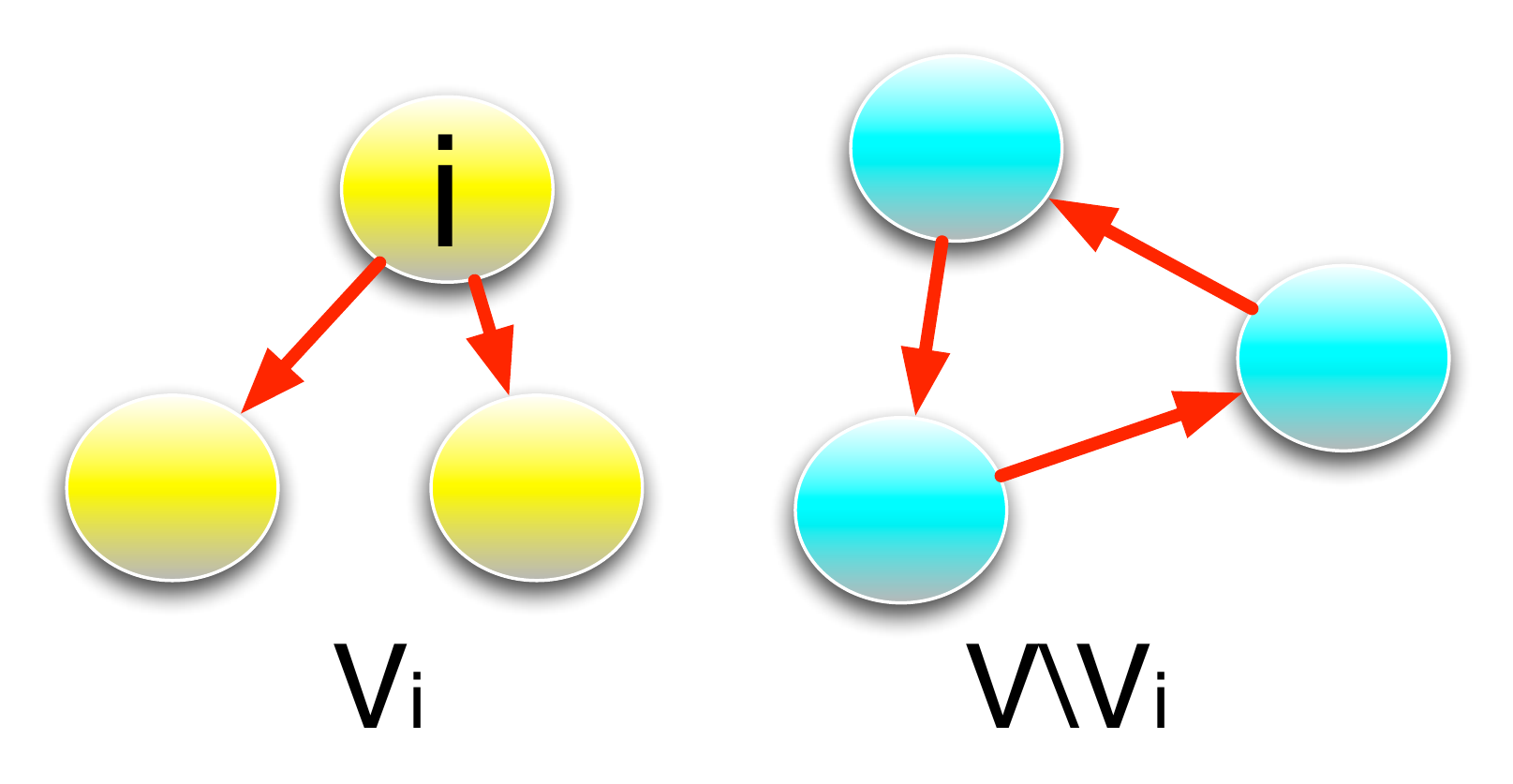}{3}
\caption{Cycles can appear at subgraph $(V\setminus V_i, E_i)$.}
\label{f.cycle}
\end{figure}
Third, \greedy{0} does not generate new cycles, but it cannot eliminate original cycles. As shown in \Cref{f.cycle}, $(V\setminus V_i, E_i)$ may contain cycles at time $0$,  which cannot be eliminated by \greedy{0}. In this section we pause discussions on cycles by assuming \Cref{a.noCycle} is valid, in the next section we will show how cycles are eliminated.

\begin{assumption}[\bf{No Cycle}]\label{a.noCycle}
At time $0$ in $(V,E)$, for any $i\in\mathcal{R}$, any node $u$ with $d_i(u)=+\infty$ does not have incoming $i$-links.
\end{assumption}
\Cref{a.noCycle} implies that no cycle exists in $(V\setminus V_i,E_i)$ for each $i\in\mathcal{R}$ at time $0$. By running \greedy{0} no $i$-links will be built beween nodes in $V\setminus V_i$ and so no cycle ever appears. The following two indicate the convergence of running \greedy{0}.

\begin{lemma}\label{lem.static}
If \Cref{a.link,a.init,a.root,a.noCycle} are valid, then at time $0$, \greedy{0}($u$,$v$) returns false for all $u,v\in V$ if and only if $|\{i\in\mathcal{R}: u\in V_i\}| = K$ for each node $u$, that is, if and only if each node is covered by $K$ trees.
\end{lemma}
\begin{proof}
If a node has $K$ incoming links with distinct labels, \greedy{0} returns false whichever target the node samples. So the if part follows.

Suppose there exists $u\in V, |\{i\in\mathcal{R}: u\in V_i\}| < K$. We prove the only if part by showing that \greedy{0} returns true when two specific nodes meet.  Node $u$ has fewer than $K$ incoming links. \Cref{a.init} implies that each node has at most $K$ incoming links. So
$|E|\leq K(N-1) + (K-1) -M\leq \sum_{u\in V}\bar{d}_u -1$ by \Cref{a.link}. Thus there exists node, say $v$, with $d(v)<\bar{d}_v$. a) if $v$ has $K$ incoming links, \greedy{0}($u$,$v$) returns true because ``Add'' can happen; b) if $v$ has fewer than $K$ incoming links, \greedy{0}($v$,$i$) returns true if $v\not\in V_i$ because ``Insert'' can happen.
\end{proof}

\begin{prop}\label{lem.cover}
Under \Cref{a.link,a.init,a.root,a.noCycle}, if \greedy{0}($u$,$v$) runs whenever $u$ samples $v$ for any $u,v\in V$, then  $|\{i\in\mathcal{R}: w\in V_i\}| = K$ for all $ w\in V$ in finite time.
\end{prop}
\begin{proof}
\Cref{a.link,a.init,a.root,a.noCycle} are valid over all time. Whenever \greedy{0} returns true, $|E|$ increases by one. But $|E|\leq KN - M$. \Cref{lem.cover} follows from \Cref{lem.static}.
\end{proof}

\greedy{0} can achieve coverage, but it has two main drawbacks: first, the depth of the tree $(V_i,E_i)$ for $i\in \mathcal{R}$ can be large; second, it cannot detect and eliminate cycles. In the next section we show how to solve these two problems by adding balance algorithms as complements.

\subsection{Achieve balance inside trees}
One way of decreasing the tree depth is to balance the tree. Here we provide a procedure under which trees can achieve balance and cycles can be eliminated. Suppose \single{0}($u$,$v$) runs whenever node $u$ samples node $v$ as a target.

\begin{algorithm}
\renewcommand{\thealgorithm}{}
\caption{\protect\single{0} (Node $u$, Node $v$)}
\label{alg.sd}

{\bf Output: } return true if $(V,E)$ changes, return false otherwise.

{\bf Tie broken: } arbitrarily

If there exists $i$ such that $u,v$ both have incoming $i$-links,
\begin{itemize}
\item{\bf Jump}:  if $d(v)<\bar{d}_v$ and $l_i^u + 1 < l_i^v$, remove $(u_{p},u)$, build $(v,u)$ in $E_i$, and return true;
\item{\bf LeafSwap}:  if $v$ is an \leaf{i} but $u$ is not, and $l_i^u > l_i^v$, remove  $(u_p,u),(v_p,v)$, build $(u_p,v),(v_p,u)$ in $E_i$, and return true.
\end{itemize}
Return false. (See  \Cref{f.single})
\end{algorithm}
\begin{figure}
\post{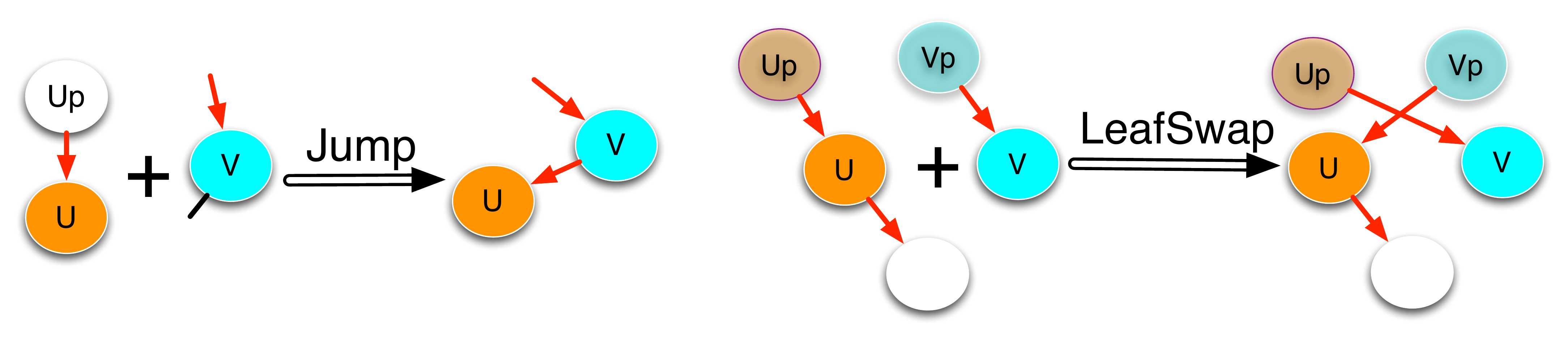}{7.5}
\caption{Single Tree Adjust}
\label{f.single}
\end{figure}
\single{p} applies the information of buffered depths, which is updated periodically but may have estimation errors. Analyzing \single{0} under depth updating is quite challenging. {\em To focus on properties of \single{0}, let us temporarily assume \Cref{a.freeError} holds. Discussion under the case wihout \Cref{a.freeError} will be covered by simulation in \Cref{sec.sim}.}

\begin{assumption}[\bf{Instantaneous distance update}] \label{a.freeError}
$\forall i\in\mathcal{R}, \forall u\in V$, assume $l_i^u =  L_i(u)$ over all time.
\end{assumption}
For any nodes $u,v\in V$, after \single{0}($u$,$v$) runs:
\begin{enumerate}
\item statements in \Cref{a.init,a.root,a.link,a.noCycle} remain valid if they are valid before running.
\item for any node, its number of incoming $i$-links and number of outgoing $i$-links both remain unchanged.
\item under \Cref{a.freeError}, for any node $w$ with $d_{i}(w)\geq1$, $L_i(w)$ is nonincreasing.
\item under \Cref{a.freeError}, no new cycles will be generated, and cycles in $(V\setminus V_i,E_{i})$ can be eliminated as shown in \Cref{f.elimCycle}, because $L_i(w)=+\infty$ for all $w\in V\setminus V_i$.
\end{enumerate}
\begin{figure}
\post{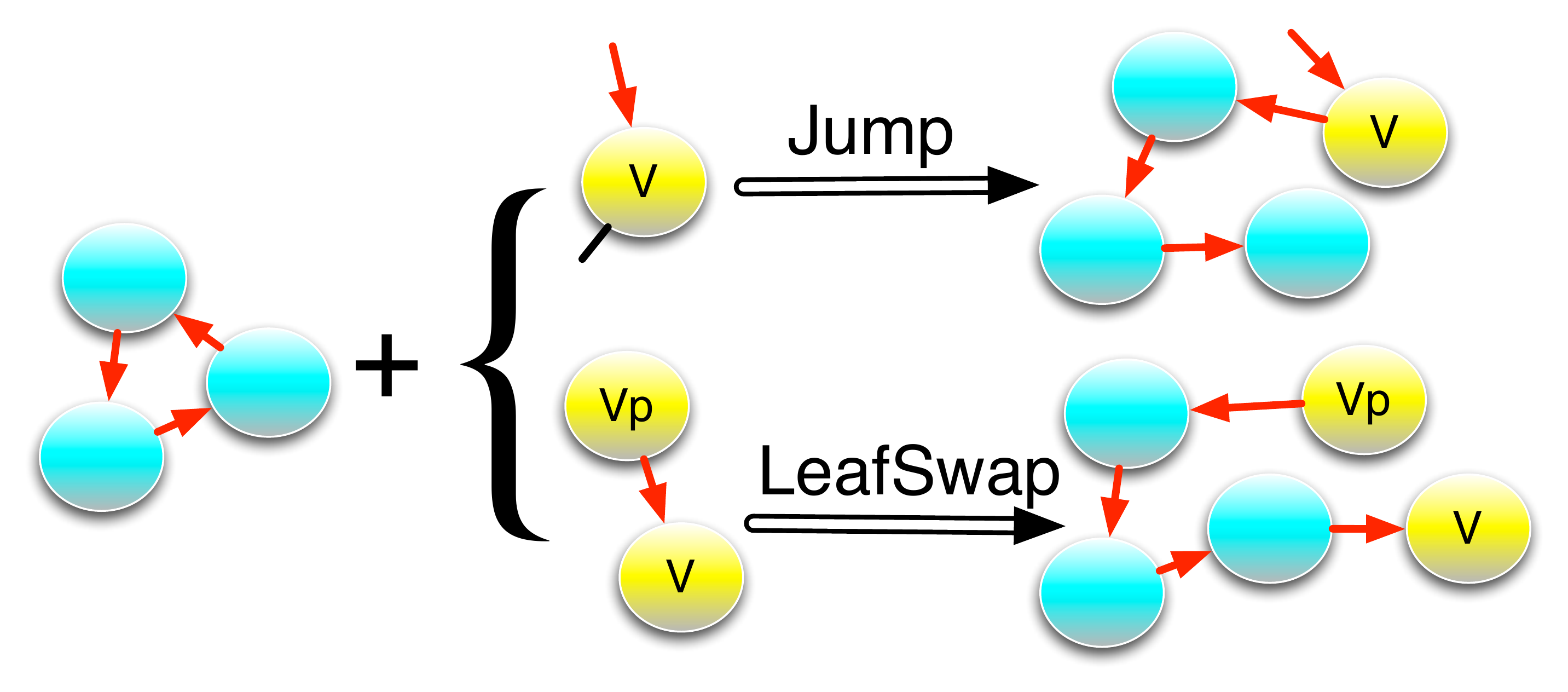}{5}
\caption{\protect\single{0} can eliminate original cycles.}
\label{f.elimCycle}
\end{figure}
The following lemma shows that \single{0} can return true unless all trees achieve balance.
\begin{lemma}\label{lem.singleBalance}
Under \Cref{a.init,a.freeError}, at time $0$, if \single{0}($u$,$v$) returns false for all pairs of nodes $u,v$, for each $i\in\mathcal{R}$, 
\begin{enumerate}
\item if \Cref{a.root} holds, \Cref{a.noCycle} also holds,
\item for any two \leafs{i} $u,v$, $|L_i(u)-L_i(v)|\leq 1$, and
\item each \internal{i} node $u$ is unavailable and $d_i(u)\geq1$.
\end{enumerate}
\end{lemma}
\begin{proof}
1) \Cref{a.root} tells that for each $i$ there is at least one \leaf{i}. Thus \Cref{a.noCycle} is valid, otherwise \single{0}($u$,$v$) returns true if $u\in V\setminus V_i$ has incoming $i$-links  and $v$ is an \leaf{i}.

2) Assume node $u$ and $v_c$ are both \leafs{i} and $L_i(u) +1 < L_i(v_c)$. Assume $v$ is the \parent{i} of $v_c$, then $L_i(u) < L_i(v)$. \single{0}($u$,$v$) returns true. Thus, depths  of any two \leafs{i} differ by at most one. 

3) If $u$ is an \internal{i} node, it cannot be an \leaf{i} because of 2), and it cannot be available otherwise \single{0}($v$,$u$) returns true where $v$ is the \leaf{i} with the largest depth to $i$.
\end{proof}

\begin{figure}
\post{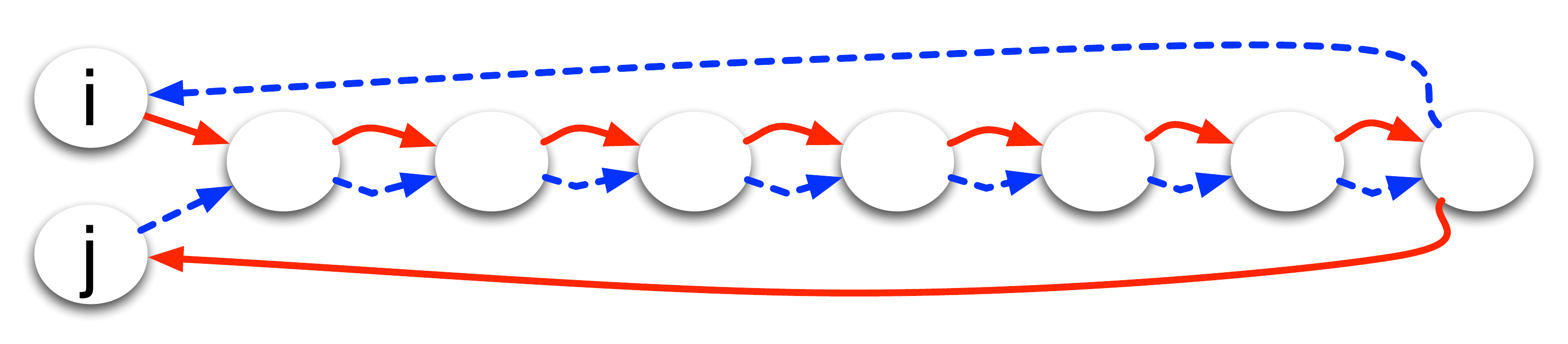}{7}
\caption{Mixed Nodes form a chain which may be very long.}
\label{f.chain}
\end{figure}
Intuitively running \greedy{0} and \single{0} together can achieve both coverage and balance, as well as eliminate cycles.  However, balance in a tree does not garantee the tree has small depth. As shown in \Cref{f.chain}, if there are many nodes with a single outgoing $i$-link, chains may appear and thereby the depth of tree $(V_i,E_i)$ can be large. Fortunately, there exists ways to eliminate conditions like that, as shown in the following.

\subsection{Achieve balance among trees}

Nodes with a single child play an important role in increasing the depth of a tree. An unavailable node $u$ with a single \child{i} for certain $i\in\mathcal{R}$ either have $\bar{d}_u = 1$ or is a mixed node. The case that $\bar{d}_u = 1$ is less interesting because most nodes can be required to provide at least two outdegrees, especially when the streaming rate of a substream is small comparing to the upload capacity. Suppose \Cref{a.2} holds for simplicity. 
\begin{assumption}[\bf{Diversity Degree}] \label{a.2}
$ \bar{d}_u \neq 1$ for all $ u\in V\setminus \mathcal{R}$.
\end{assumption}
Notice that in \Cref{a.2} existence of nodes with outdegree zero is allowed. \Cref{a.2} gets rid of the case that many nodes have outdegree one.
Reducing the number of mixed nodes can lower the tree depths. In this section we provide \mixnode{p} under which the number of mixed nodes can be greatly reduced. For any pair of nodes $u_c,v$, suppose \mixnode{0}($u_c$,$v$) runs when $u_c$ samples $v$ as a target.
\begin{algorithm}
\renewcommand{\thealgorithm}{}
\caption{\protect\mixnode{0} (Node $u_c$, Node $v$)}
\label{alg.mix}

{\bf Output: } return true if $(V,E)$ changes, return false otherwise.

{\bf Tie broken: } arbitrarily.

If there exist $i,j,i\neq j$ such that $v$ has an \child{j} say $v_c$, $u_c$ has an \parent{i} say $u$, and 

~ {\bf MixSwap}: if $l_i^u \geq l_i^v, l_j^u \leq l_j^v$ and either of the two is true:
\begin{itemize}
\item[a)] $l_i^u\neq l_v^i$ or $l_j^u\neq l_j^v$,
\item[b)]  $l_i^u= l_i^v$, $l_j^u = l_j^v$, $(u-v)(j-i) > 0$, (Note $u,v,i,j$ are ids in \{1,...N\}.)
 \end{itemize}
then remove $(u,u_{c}),(v,v_{c})$,  build $(u,v_{c})$ in $E_j$, build $(v,u_{c})$ in $E_i$, and return true.

Otherwise return false. (See \Cref{f.mixadjust})
\end{algorithm}
\begin{figure}
\post{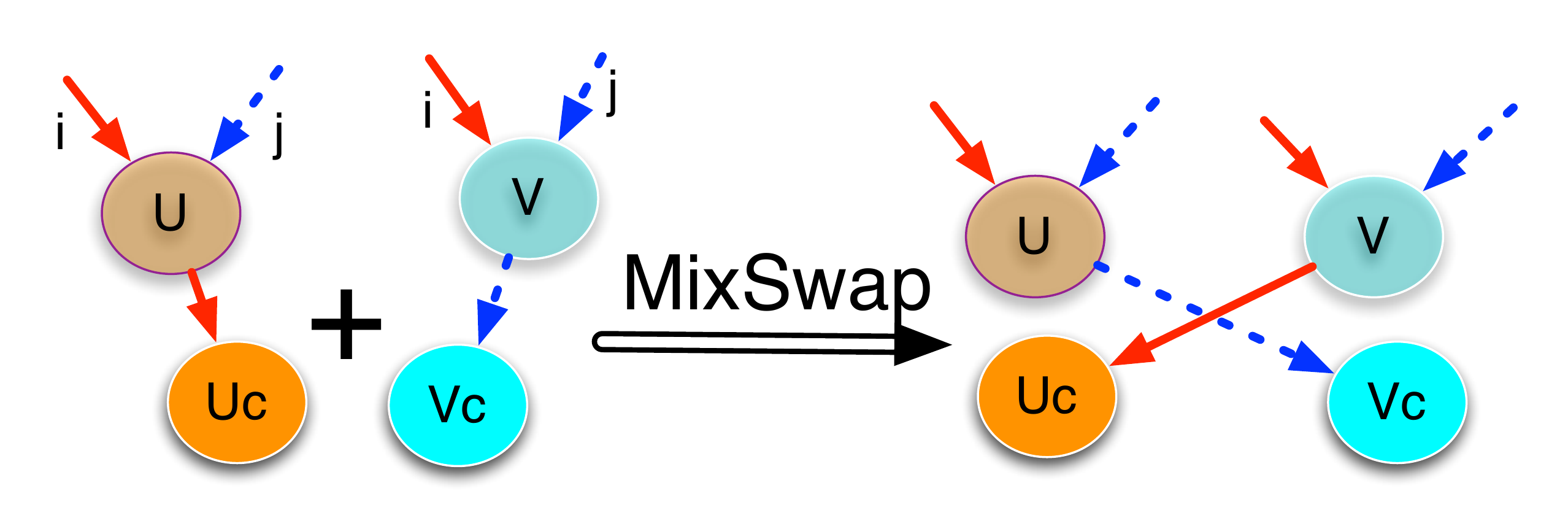}{5.5}
\caption{$u_c$ and $v_c$ switch their parents if one can decrease its depth while the other one's depth does not increase, or depths are unchanged but lower-id parents can get lower-color links.}
\label{f.mixadjust}
\end{figure}

\mixnode{0}($u_c$,$v$) returns true if after ``MixSwap'', between $u_c$ and one child of $v$, one can decrease its depth while the other one's depth does not increase, or depths are unchanged but parents with lower ids can get outgoing links with lower colors. We break the tie by assuming parents with lower ids have a preference on links with lower colors, so as to eliminate the case that there are many mixed nodes with exactly the same depths in multiple trees.

Under \Cref{a.freeError}, for any pair of nodes $u,v$, after \mixnode{0}($u$,$v$) runs,
\begin{enumerate}
\item statements in \Cref{a.init,a.root,a.link,a.noCycle} remain valid if they are valid before running.
\item $L_i(w)$ is nonincreasing for any node $w$,
\item no new cycle is generated.
\end{enumerate}
Moreover, the following lemma indicates that \mixnode{0} can return true unless depth vectors of mixed nodes form a strict chain:
\begin{lemma}\label{lem.mix}
Under \Cref{a.freeError}, if \mixnode{0}($\bar{u}$,$\bar{v}$) returns false for all pairs of nodes $\bar{u},\bar{v}\in V$, for any $i,j\in\mathcal{R}$ and any two \mix{i}{j}s $u,v$, either $(L_i(u),  L_j(u)) < (L_i(v), L_j(v))$ or $ (L_i(v), L_j(v))< (L_i(u),  L_j(u)) $.
\end{lemma}
\begin{proof}
The lemma follows by noticing that if $u,v$ are both \mix{i}{j}s, ``MixSwap'' can happen when either a child of $u$ contacts $v$ or a child of $v$ contacts $u$, unless  $(L_i(u),  L_j(u))$ and $(L_i(v), L_j(v))$ are in a strict order.
\end{proof}

\Cref{lem.mix} still does not guarantee that the number of mixed nodes is small. As shown in \Cref{f.chain}, mixed nodes can form a long chain where the conclusion in \Cref{lem.mix}  still holds. However, the appearances of the structure in \Cref{f.chain} are quite rare because random sampling is assumed. And intuitively it becomes rarer as $N$ increases.
In practice we may safely ignore it, but here for completeness of analysis, we eliminate the possibility of a long chain as in \Cref{f.chain} by assuming \Cref{a.hat} holds. We will show later by simulation that ignoring \Cref{a.hat} does not harm performance.

\begin{assumption}[\bf{Shower head}]\label{a.hat}
At time $0$, there exists $c\in \mathcal{Z}^+$ such that for each $ i\in \mathcal{R}$, there are at least $M$ \leafs{i} in $(\hat{V}_i, E_i)$, where $\hat{V}_i: = \{ u\in V : d_i(u)\leq c\}$.
\end{assumption}
\Cref{a.hat} says that initially in any tree the subtree of nodes with depth bounded by $c$ has at least $M$ leaves.
Intuitively \Cref{a.hat} suggests there is something analogous to a shower head, which provides enough branches near the top of each tree $(V_i,E_i)$. The value $c$ in \Cref{a.hat} can be as small as $O(\log M)$, or even $O(1)$ if the root or its children have large outdegrees.
We have \Cref{lem.balance} under \Cref{a.hat}.

\begin{lemma}\label{lem.balance}
Under \Cref{a.init,a.freeError,a.2,a.hat}, if \single{0}($u$,$v$) and \mixnode{0}($u$,$v$) both return false for all $u,v\in V$, then for each $i\in\mathcal{R}$ the depth of each tree $(V_i,E_i)$ is less than or equal to $\log_2(N+1) + c$, where $c$ is defined in \Cref{a.hat}.
\end{lemma}
\begin{proof}
Suppose the depth of tree $(V_i,E_i)$ is $\bar{l}_i$. \Cref{lem.mix} tells that for each  $ j\neq i$ and for each $k\leq \bar{l}_i$, there are at most one \mix{i}{j}s whose depth to $i$ is $k$. So there are at most $M-1$ mixed nodes which have \child{i} and whose depth to $i$ is $k$.  
\Cref{lem.singleBalance} tells that all \internal{i} nodes must have at least $2$ children because they are unavailable and because \Cref{a.2} holds.

Thus, for each $k\leq \bar{l}_i-2$, in the set of nodes whose depth to $i$ is $k$, at most $M-1$ of them can have a single \child{i}, while each of the other nodes has at least $2$ $i$-children because they are unavailable non-mixed \internal{i} nodes. 

Notice that there are at least $M$ nodes whose depth to $i$ is $c$.   Thus, the number of nodes whose depth to $i$ is in $[c, \bar{l}_i -1]$ is at least $1 + 2 + 4 + ...+ 2^{\bar{l}_i-1- c}\leq N$, so $\bar{l}_i\leq \log_2(N+1) + c$.
\end{proof}

\subsection{Combine everything together\label{sec_whole}}

Here we combine all parts above together.
For each pair of nodes $u,v$, run \algwhole{0}($u$,$v$) when $u$ samples $v$.

\begin{algorithm}
\renewcommand{\thealgorithm}{}
\caption{\protect\algwhole{0} (Node $u$, Node $v$)}
\label{alg.whole}

{\bf Output: } return true if $(V,E)$ changes, return false otherwise.

{\bf Tie broken: } arbitrarily

\begin{algorithmic}  
\STATE Nodes $u,v$ update their buffered depth by running \depUpdate{0} for each $i\in\mathcal{R}$ respectively
    \RETURN \greedy{0}($u$, $v$) \OR \\
    ~~ \single{0}($u$, $v$) \OR
    \mixnode{0}($u$, $v$)
\end{algorithmic}
\end{algorithm}
Just like that in C or C++, in \algwhole{0},  if operation ``A'' returns true, operation ``A or B'' immediately returns and operation ``B'' does not run.

Notice that buffered depths are also updated whenever new links are built, as assumed in \Cref{sec.dDist}. And
Notice that under \Cref{a.freeError}, for any pair of nodes $u,v$, after \algwhole{0}($u$,$v$) runs,  statements in \Cref{a.init,a.link,a.root,a.noCycle,a.2,a.hat} remain valid if they hold before running, respectively.

\Cref{lem.balanceconv} shows that by running \algwhole{0} certain metric changes monotonely. Some additional definitions are necessary before defining the metric.
Define $L'_i(u) = \min\{L_i(u), N\}$, i.e., $L'_i(u)$ is the same as $L_i(u)$ except that $L'_i(u)=N$ instead of $+\infty$ if there is no path from $i$ to $u$.
Let $Y= \sum_{u\in V}\sum_{i\in\mathcal{R}} L'_{i}(u)$ so $Y$ is the sum of all modified depths of all nodes.
Let $D(u):=\sum_{i\in\mathcal{R}} i \cdot d_{i}(u)$ be the sum of the colors of all outgoing links of $u$, and let $S: = \sum_{u} u D(u)$.
\begin{lemma}\label{lem.balanceconv}
Under \Cref{a.freeError}, for any pair of nodes $u,v$, after \algwhole{0}($u$,$v$) runs, if it returns true, $(-|E|, Y, -S)$ decrease lexicographically by at least one, otherwise $(-|E|, Y, -S)$ does not change.
\end{lemma}
\begin{proof}
Under \Cref{a.freeError}, if \greedy{0} returns true, $|E|$ increases by one; if \single{0} returns true, or \mixnode{0} returns true because of Condition (a), $Y$ decreases by at least one but $|E|$ remains unchanged; if \mixnode{0} returns true because of Condition (b), $S$ increases by at least one while $|E|,Y$ remain unchanged.
\end{proof}
\Cref{lem.balanceconv} helps to show  convergence of the algorithm.

\begin{prop}\label{lem.conv}
Under \Cref{a.init,a.link,a.root,a.freeError,a.2,a.hat}, suppose \algwhole{0}($u$,$v$) runs whenever $u$ samples $v$ for any $u,v\in V$. Then in finite time \algwhole{0}($u$,$v$) returns false for all $u,v$, and at that time,
\begin{enumerate}
\item[(a)] $\forall w\in V, |\{i\in\mathcal{R}: w\in V_i\}| = K$,
\item[(b)] $\forall i\in\mathcal{R},$ the depth of the tree $(V_i,E_i)$ is bounded by $\log_2(N+1) + c$, where $c$ is the value in \Cref{a.hat}.
\end{enumerate}
\end{prop}
\begin{proof}
Notice that statements in \Cref{a.init,a.link,a.root,a.2,a.hat} are valid over all time. \Cref{lem.balanceconv} and the boundness of $|E|, Y$ and $S$ tell that in finite time \algwhole{0} will return false whenever any two nodes meet. \Cref{lem.static} tells that (a) is valid while \Cref{lem.balance} tells that (b) is valid.
\end{proof}

\subsection{Comment on distributed depth update}
Under \Cref{a.freeError}, no new cycles can appear by running \algwhole{0}, which is not the case when depths are updated distributedly. Here we argue that cycles are rare and can be eliminated quickly even if \Cref{a.freeError} does not hold.

First, a new cycle appears only if a node builds an incoming link from one of its descendants, which happens only if the descendant has a smaller buffered depth. That condition is rare because 1) most nodes just have several descendants; 2) if $u$ is a descendant of $v$ and if the depth of $u$ is larger than the depth of $v$, by running \single{0} and \mixnode{0} the depth of $u$ can remain larger than that of $v$. Usually a node has smaller depth than its ancestors only when many nodes suddenly become ancestors of the node in a short time, which is also a rare event.

Second, even if a new cycle appears, it will disappear in a short time. By \depUpdate{0}, depths of nodes in a cycle keep updating and can count to a large value, just like the ``counting to infinity'' problem in network routing. Whenever a node with a large depth meets a leaf node, changes as shown in \Cref{f.elimCycle} can happen and thereby the cycle disappears. At least half of the nodes in a tree are leaves, so by random sampling, cycles are eliminated quickly.

In summary, we argue that distributed depth update does not harm performance much compared to that under \Cref{a.freeError}, which is supported by simulations in \Cref{sec.sim}.

\section{ Rate of Convergence}\label{sec_converge}

Analyzing the convergence time of running \algwhole{0} is highly challenging.  
This section provides a stochastic bound for the convergence time under the case of a single tree, i.e, $M=K=1$. We further assume that each node has outdegree at least $2$, i.e., $\bar{d}_u\geq2$ for all $u\in V$. And suppose \Cref{a.freeError} holds so that each node knows its depth to the root. For theoretical tractability, instead of running \algwhole{0}, we simply the algorithm by assuming that only ``Add'' in \greedy{0} and ``Jump'' in \single{0}  run when two nodes meet.

The simplified model is summarized as follows: each node knows its depth to the root; whenever a node's Poisson clock ticks, the node samples a target uniformly at random, if the target is available and the depth of the target is less than the depth of the node by at least two, the node removes its current incoming link if there is and builds a new incoming link from the target; otherwise nothing happens. Assume initially each node has at most one incoming link, then convergence time is upper bounded:
\begin{prop}\label{p.rate}
Let $T$ be the first time for the maximum depth of all nodes to be bounded up by $\cei{\log_2(N+1)}$, then
$$
\forall \epsilon > 0, P\left[T> 21\log_2(N+1) + 16\epsilon\right] < 3e^{-\epsilon}.
$$
\end{prop}

Notice that the maximum depth is $+\infty$ if there is a node to which no path exists from the root. So \Cref{p.rate} bounds the time for the tree to cover all nodes and achieve balance. It implies that the model converges in $O(\log N)$ time.

We argue that for the case of multiple trees similar bounds as in \Cref{p.rate} can be generated, and by running \algwhole{0} the network can converge in $O(\log N)$ time. Because when targets are unavailable, ``LeafSwap'' substitutes ``Jump''  efficiently since half nodes are leaves, and ``Insert'' substitutes ``Add'' as well.

The proof of \Cref{p.rate} is provided below.

\subsection{proof of \Cref{p.rate}}

We assume nodes sample targets randomly at times of Poisson processes with rate $1$. Equivalently, we can assume that each $\langle node, node\rangle$ pair maintains a Poisson clock which ticks at rate ${1\over N}$. The following definitions are applied for the proof.

\begin{defi}
Define $l_f := \flo{\log_2(N+1)} -1$ and $l_c := \cei{\log_2(N+1)}-1$. 

Define $l_\alpha := \cei{\log_2\left((1-\alpha) N+1\right)}-1$.

Define $Z_i, i\geq 0$ to be the number of nodes with depths  $\leq i$. Note that $Z_i(t)$ is a discrete counting process. 

Define $Z_{-1} = \tZ_{-1} \equiv 0$.
\end{defi}

The model describes a Markov process with state being $G=(V,E)$. We apply $G=(V,E)$ to denote the process as well as the graph.
It is not difficult to see that graph $G=(V,E)$ can converge to a balanced tree covering all nodes in finite time, because 1) the depth of each node is nonincreasing, i.e., $Z_i$ for each $i$ is nondecreasing, and 2) there exist nodes which can decrease their depths if $(V,E)$ is not balanced or $(V,E)$ does not cover all nodes.

The process is separeted into two phases, illustrated by \Cref{lem.T0,lem.T1}, respectively. The time for the first phase is described below.

\begin{lemma}\label{lem.T0}
For any $ \alpha\in(0,1)$, 
let $T_0$ be the first time that $Z_{l_{\alpha}} \geq (1-\alpha) N$,
$$
P[T_0 >  t] \leq 2^{l_\alpha+1} P[Poi(\alpha t/2)\leq l_{\alpha}-1].
$$
\end{lemma}
\begin{proof}
Define an alternative process $G' = (V',E')$ such that it is identical to the original process $G$ when $t< T_0$; when $t \geq T_0$, whenever a node with depth $>l_\alpha$ changes its depth, a new node with fan-out degree $2$ whose depth is $\infty$ arrives to $G'$. After $T_0$, $|V'|$ may increase but we assume the Poisson clock of each $(node,node)$ pair still ticks at rate $\mu=1/N$. On $G'$, the number of nodes with depths $>l_\alpha$ does not change after $T_0$, and always $\geq \alpha N$. The probability $P[t<T_0]$ is identical for processes $G$ or $G'$. 
{\em In the following, our discussion are on process $G'$. For simplicity, we apply the same notations for $G'$ as for $G$}. Let $Z=(Z_0(t),Z_1(t),...Z_{l_\alpha}(t))$.

Notice that the number of available nodes with depths $\leq i$ is larger or equal to $(1+2Z_{i-1}-Z_{i})^+/2$: consider each node labels its outgoing degrees and marks the first two degrees red. Each node has at least $2$ red degrees but there are only $Z_{i}-1$ nodes to serve. So there are at most $(Z_{i}-1)/2$ nodes whose red degrees are both taken. The number of nodes with depths $>l_\alpha$ is larger or equal to $\alpha N$. Thus, if $i\leq l_\alpha$, the transition rate for $Z_i$ to jump is lower bounded by 
$$
\mu\alpha N (1+2Z_{i-1}-Z_{i})^+/2 = {\alpha \over 2} (1+2Z_{i-1}-Z_{i})^{+}.
$$ 

There exists a process $\tZ=(\tZ_0(t), \tZ_1(t),...\tZ_{l_\alpha}(t))$ in $\mathcal{Z}_+^{l_\alpha+1}$ on an extended probability space such that each coordinate of $\tZ$ has jumps of size one and  jump rate for $\tZ_i$ is ${\alpha } (1+2\tZ_{i-1}-\tZ_{i})^{+}/2$. Notice that simultaneous jumps of different coordinates of $\tZ$ are allowed. Let $\tZ(0)=(1,1,1,...1)$. Initially we have $Z(0)\geq \tZ(0)$.

Process $Z$ and $\tZ$ can be coupled so that $Z(t)\geq \tZ(t)$ with probability one for all $t$. That is because if $Z\geq\tZ$ then jump rate of $Z_i$ is larger or equal to jump rate of $\tZ_i$ for all $i$ such that $Z_i = \tZ_i$. So the jumps of $\tZ$ can be obtained by generally thinning the jumps of $Z$, and adding more jumps to $\tZ_i$'s with $Z_i>\tZ_i$.

By induction it is easy to show that $\tZ_i(t)\leq 2^{i+1}-1$ with probability one for all $t$, because jump rate of $\tZ_i$ is zero if $\tZ_i= 2^{i+1}-1$. Moreover, $\forall  i\in[0, l_\alpha]$, 
\begin{eqnarray*}
{d E[\tZ_i(t)]\over dt} &=& E\left[{\alpha \over 2} \left(1+2\tZ_{i-1}(t)-\tZ_{i}(t)\right)^{+}\right]\\
&\geq& {\alpha\over 2}\left(1 + 2E[\tZ_{i-1}(t)]  - E[\tZ_i(t)]\right)^{+}.
\end{eqnarray*}

Let $\mathbf{y}(t) = (y_0(t), y_1(t),...y_{l_f-1}(t))$ to be a motion trajectory defined by the following linear differential equation:
\begin{eqnarray*}
&& y_i'(t) = {\alpha\over 2}\left(1 + 2y_{i-1}(t) - y_i(t)\right)^{+},\\
&&y_i(0) =1, y_{-1}(t) = 0, 0\leq i\leq l_\alpha
\end{eqnarray*}
By induction it is easy to prove that, $E[\tZ_i(t)] \geq y_i(t)$.

Define $\Delta(t) := \left({2^{l_\alpha+1}-1} - \tZ_{l_\alpha}(t)\right)$.
Notice that $\Delta(t)\geq 0$, apply Markov's inequality and \Cref{lem.diff}, 
\begin{eqnarray*}
&&P[T_0>t] = P[Z_{l_\alpha}(t) \leq (1-\alpha)N-1] \\
&\leq& P[\tZ_{l_\alpha}(t) \leq 2^{l_\alpha+1}-2] = P[\Delta(t) \geq 1] \\
&\leq& E[\Delta(t)] \leq 2^{l_\alpha+1}-1 - y_{l_\alpha}(t)\\
&\leq& 2^{l_\alpha+1} P[Pois(\alpha t/2)\leq l_{\alpha}-1].
\end{eqnarray*}
The lemma follows. 
\end{proof}
\begin{lemma}\label{lem.diff}
Let $\mathbf{y}(t) = (y_0(t), y_1(t),...y_{k}(t))$ to be a motion trajectory defined by the following:
\begin{eqnarray*}
&&y_i'(t) = \beta\left(1 + 2y_{i-1}(t) - y_i(t)\right)^{+}, s.t\\
&&y_{-1}(t) = 0, y_0(0) = 1, y_i(0) \geq 0, 0\leq i\leq k,
\end{eqnarray*}
where $\beta>0$ is a constant, then
$$y_i(t)\geq 2^{i+1}\left\{1-P[Pois(\beta t)\leq i-1]\right\} -1.$$
\end{lemma}
\begin{proof}
Let $\delta_i(t): = \left[{2^{i+1}-1} - y_i\left({t/\beta}\right)\right] / 2^{i+1}$, we can simplify the equation of $y$ as
\begin{eqnarray*}
\delta_i'(t) = \left(\delta_{i-1}(t) - \delta_i(t)\right)^{-}.
\end{eqnarray*}
Notice that by induction we can show that $\delta_i(t)\leq \theta_i(t)$, where
$$
\theta_i(t) = \theta_{i-1}(t) - \theta_i(t),\theta_i(0) = \delta_i(0).
$$
Solving the differential equation about $\theta$ gives that
$$
\delta_i(t)\leq \theta_i(t) = e^{-t}\sum_{k=0}^{i-1} \theta_{i-k}(0) {t^k\over k!}.
$$
$
\forall i, \theta_i(0) ={{2^{i+1}-1} - y_i\left({0}\right) \over 2^{i+1}} \leq 1.
$
Thus
\begin{eqnarray*}
2^{i+1}\delta_i(t)&\leq& 
2^{i+1}\theta_i(t)= 2^{i+1}P[Pois(t)\leq i-1],
\end{eqnarray*}
and so the lemma follows.
\end{proof}

The time of the second phase is described below.

\begin{lemma}\label{lem.T1}
For any $\alpha\in(0,0.5)$, 
given $Z_{l_\alpha}(0) \geq (1-\alpha)N$, let $T_1$ be the first time that $Z_{l_c+1} = N$, then
$$
P[T_1 \leq t] \geq \left[1 - e^{-(1-2\alpha)t/2}\right] ^{\alpha N}.
$$
\end{lemma}
\begin{proof}
Let $X_i$ be the number of nodes with depths $=i$, and let $Y_i$ be the number of available nodes with depths $\leq i$.

Consider the jumping rate of $Z_{l_c + 1}$. Notice that $Z_{l_\alpha}(t)\geq (1-\alpha) N$, and
\begin{eqnarray*}
&&2Y_{l_c} \geq \max_{ i\in [l_\alpha, l_c]}\left\{1+Z_{i}-X_{i+1}\right\}\\
&\geq& \max_{ i\in [l_\alpha, l_c]}\left\{  (1-\alpha)N + \sum_{k=l_\alpha+1}^i X_k - X_{i+1} \right\}\geq (1-2\alpha) N.
\end{eqnarray*}

The last inequalities above is due to the fact that $\sum_{k=l_\alpha+1}^{l_c+1} X_k \leq \alpha N$. 

The rate for any node with depth $>l_c+1$ to jump to join $Z_{l_c+1}$ is at least $\mu Y_{l_c} \geq (1-2\alpha)/2$. So the lemma follows.
\end{proof}

Now \Cref{lem.T0,lem.T1} are combined to prove \Cref{p.rate}.
Consider \Cref{lem.T0}, apply the Chernoff bound for Poisson variable: $P[Poi(\lambda) \leq x] \leq {e^{-\lambda}(\lambda e)^x\over x^x}$ if $\lambda > x$, and $l_\alpha -1 <l_\alpha \leq \log_2(N+1)$, we have
$$
P[T_0>2t/\alpha]\leq 2\exp\left\{ r(1+ln 2 - {t\over r} + ln{t\over r}) \right\},
$$
where $r= \log_2(N+1)$. Notice that $ 1+ln 2 - {t\over r} + ln{t\over r}\leq -0.2 - 2(t/r-3)/3 = - 2(t/r-2.8)/3$, so
$$
P\left[T_0>{2\over \alpha}\left(2.8\log_2(N+1) +{3\over2}\epsilon\right)\right]\leq 2e^{ -\epsilon }.
$$
Consider \Cref{lem.T1}:
$$
P[T_1 \leq t] \geq 1 - \alpha N e^{-(1-2\alpha)t/2} \geq 1 -  (N+1) e^{-(1-2\alpha)t/2}.
$$
And so
$$
P\left[T_1 > {2\big(ln (N+1) + \epsilon\big)\over 1-2\alpha}\right] \leq e^{-\epsilon}.
$$
Choose $\alpha = 0.36987$, which minimizes $5.6/\alpha + 2\ln2/(1-2\alpha)$, we get
$$
P\left[T >21\log_2(N+1) +16\epsilon\right]\leq 3e^{-\epsilon}.
$$


\section{Simulation}\label{sec.sim}

We show that \algwhole{0} works pretty well under \Cref{a.init,a.root,a.2,a.link}, without \Cref{a.hat,a.noCycle,a.freeError}.
Let each node sample targets randomly with rate $\mu=1$ and run \algwhole{0}. \Cref{a.freeError} is not invoked so depths update distributedly. Notice that \algwhole{0} runs instantaneously in simulations. In each experiment below,  we set $N=1000$ fixed; at time $0$, we first set $E$ to be empty, then let each root $i$ build an \child{i} which is randomly selected from $V\setminus\mathcal{R}$. So at time $0$ $E$ contains $|\mathcal{R}|$ links and \Cref{a.init,a.root} both hold. We let \Cref{a.2,a.link} hold too.

Because \Cref{a.init} holds, during each simulation below $(V_i,E_i)$ for each $i$ is always a tree. Tree $i$ is given by $(V_i,E_i)$. Say that a node is {\em covered by tree $i$} if $L_i(u) < +\infty$, and say that a node is {\em fully covered} if it is covered by at least $K$ trees. 

In experiments below, the parameters chosen include $K,M$ and the degree vector $(\bar{d}_u)_{u\in V}$. For each selection of parameters $K,M,(\bar{d}_u)_{u\in V}$, we repeat running the experiment $500$ times, with each experiment running for $100$ time units and with system states recorded in the same time units. Metrics considered include the fraction of nodes fully covered and the maximum tree depth.

\subsection{Homogenous degrees and tight capacity constraint}
In this series of experiments, we let each root have degree $K-1$: $\forall i\in\mathcal{R},\bar{d}_i = K-1$, and let each non-root node have degree $K$: $\forall u\in V\setminus\mathcal{R}$, $\bar{d}_u = K$. The capacity is tight because the equality in \Cref{a.link} is achieved: $\sum_{u\in V}\bar{d}_u = KN-M$. Keep $K\geq 2$ so \Cref{a.2} holds. After repeating an experiment for $500$ times, for $a=0.2,1,5,50,100$,  we record metrics of the $a\%$ worst experiment at each time $t$. Notice that each line with legend ``worst case'' correponds to the case $a=0.2$, which means that there is no experiment performing worse than the line at any time. ($0.2\% * 500 = 1$).

\begin{figure}
\begin{subfigure}[l]{0.22\textwidth}
	\post{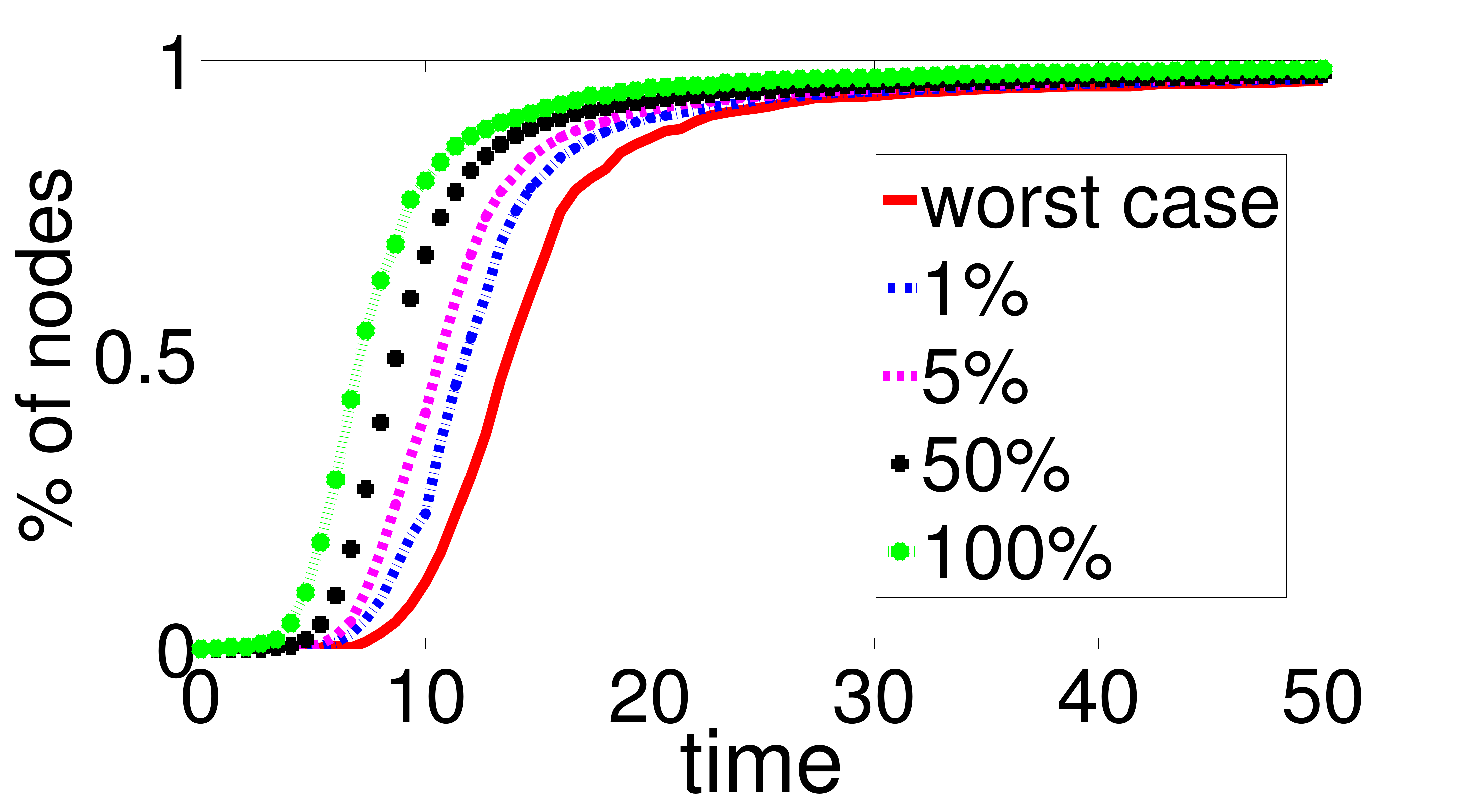}{4.3}
	\caption{$\%$ of nodes fully covered}
    \label{f.k2number}
\end{subfigure}
\begin{subfigure}[r]{0.22\textwidth}
	\post{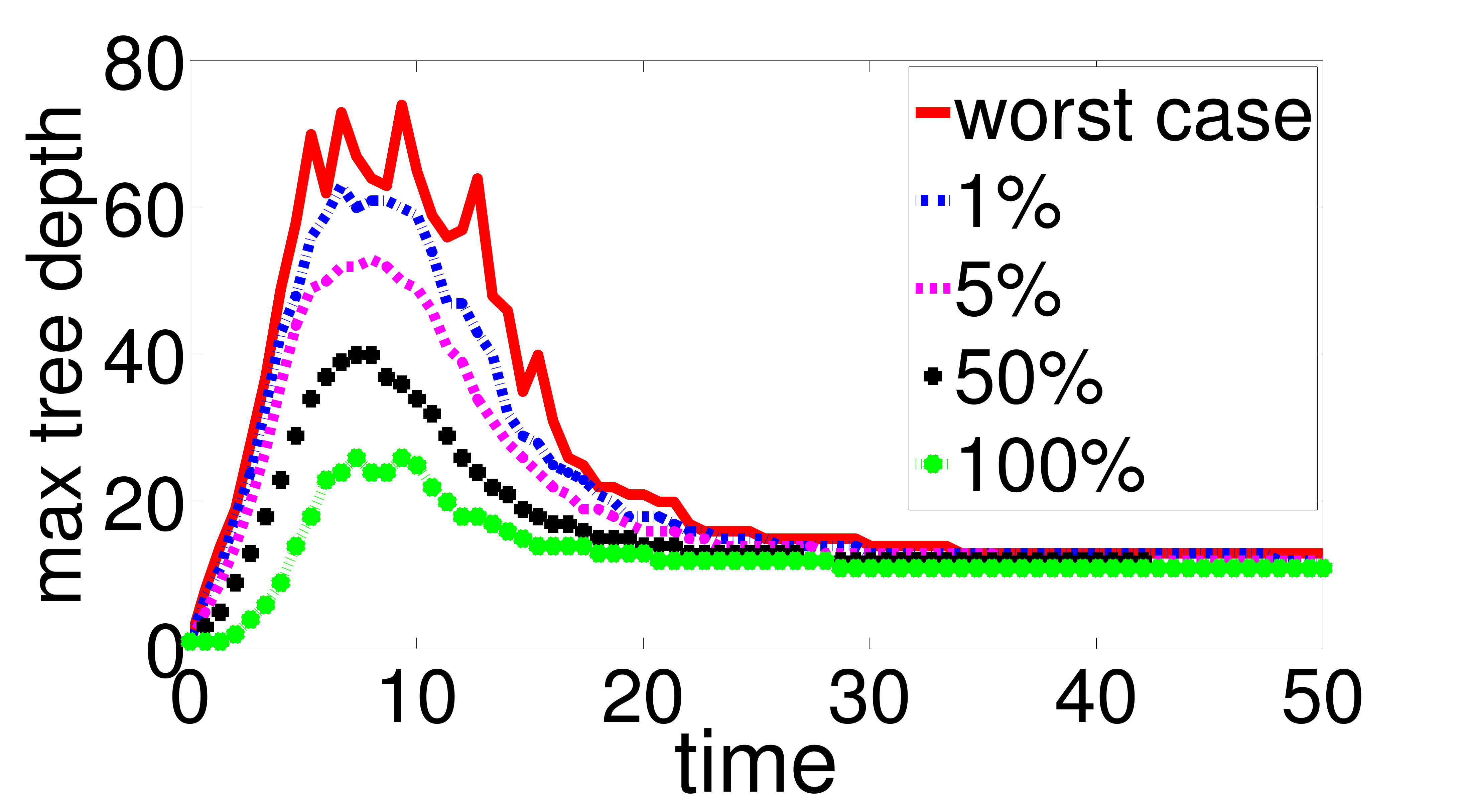}{4.3}
 	\caption{max tree depth}
\label{f.k2tree}
\end{subfigure}
\caption{Cumulative distribute when $M=K=2$. Point $(t,y)$ on a line legended $a\%$ means only in $a\%$ of $500$ experiments the corresponding metric at time $t$ is worse than $y$.}
\label{f.K2}
\end{figure}

For example, in \Cref{f.K2} we set $M=K=2$. A point $(t,y)$ in \Cref{f.k2number} on the line with legend ``$5\%$''  means that in $5\%$ of $500$ repeated experiments, the fraction of nodes fully covered at time $t$ is no larger than $y$; a point  $(t,y)$ in \Cref{f.k2tree} on the line with legend ``$1\%$''  means that in $1\%$ of $500$ repeated experiments, the max tree depth at time $t$ is no less than $y$. 

\Cref{f.K2} shows what a specific sample path looks like under \algwhole{0}. In \Cref{f.k2number}, we can see that the fraction of nodes fully covered increases almost exponentially from $0$ to $1$, over $90\%$ nodes are fully covered by time $25$ under $99\%$ experiments. That is because nodes can gradually increase the number of trees covering them until they get fully covered, as they meet other fully covered nodes.
{\em It appears that the fraction of nodes is almost nondecreasing over all time, which validates that cycles generated are rare and are quickly eliminated.} 
 As indicated in \Cref{f.k2tree}, the maximum tree depth increases linearly in the beginning, then decreases almost exponentially, and finally converges to below $12$. At time $25$, in $99\%$ repeated experiments max tree depths are below $20$. The rate of convergence follows \Cref{p.rate}, though \Cref{p.rate} is for the case of one tree.

Notice that not only in \Cref{f.K2}, but also in all our simulations below, the ``worst case'' lines are quite close to the ``1\%'' lines, but the latters are much more smooth than the former.
{\em In the following, we apply ``1\%'' lines instead of ``worst case'' lines to describe performance.}

\begin{figure}
\begin{subfigure}[l]{0.22\textwidth}
	\post{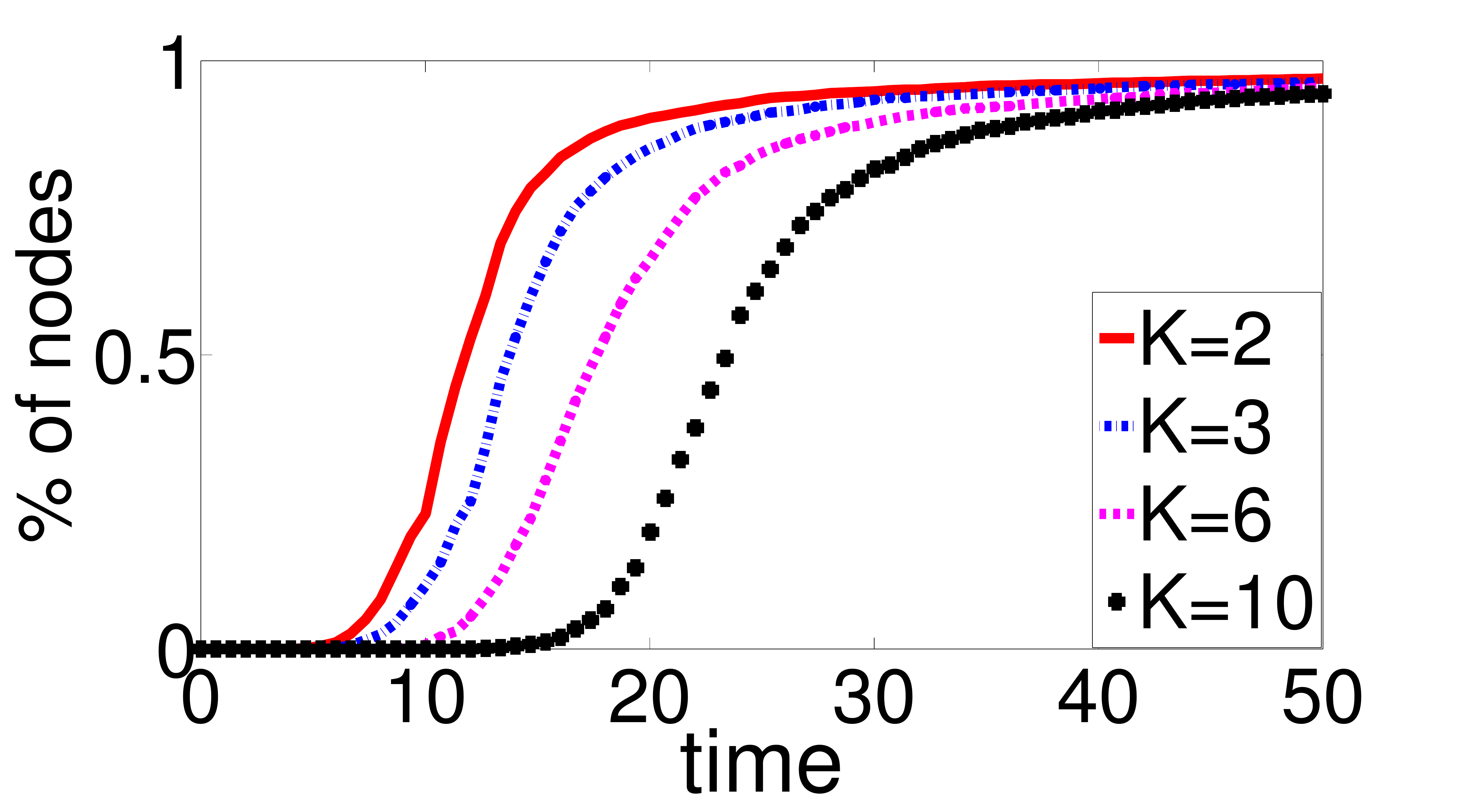}{4.3}
	\caption{$\%$ of nodes fully covered}
    \label{f.k2_10number}
\end{subfigure}
\begin{subfigure}[r]{0.22\textwidth}
	\post{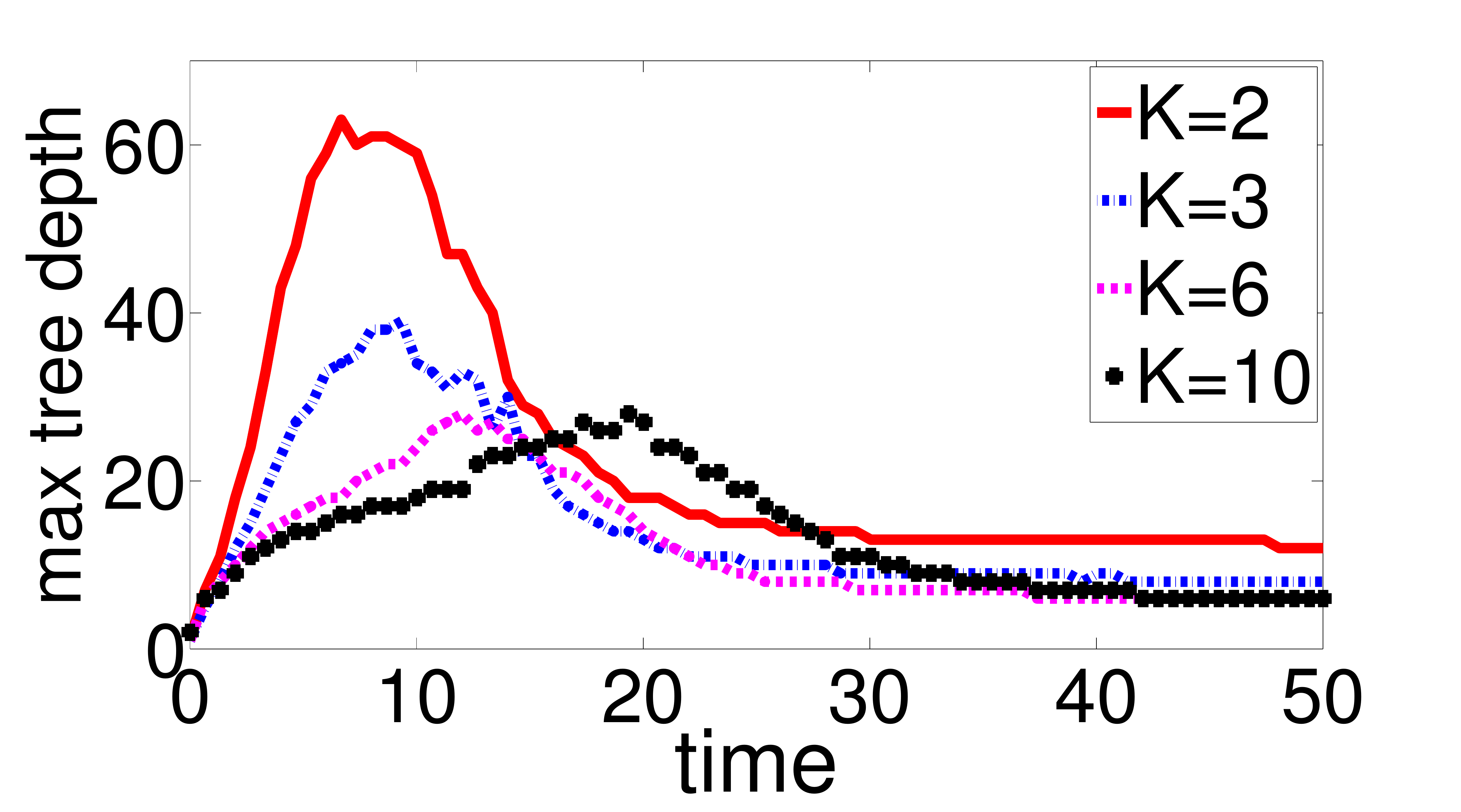}{4.3}
 	\caption{max tree depth}
\label{f.k2_10tree}
\end{subfigure}
\caption{$1\%$ cumulative lines when $M=K$ varies. }
\label{f.K2_10}
\end{figure}

In \Cref{f.K2_10} we test different $K$'s so as to make sure convergence follows in other cases. Notice that typically in practice $K$ is below $10$. Set $M=K$ with $K$ varying in $\{2,3,6,10\}$. We expect to observe similar images as in \Cref{f.K2}, with longer convergence time as $K$ increases because each node has to get covered by more trees. For each $K$, we draw the $1\%$ worst case line in \Cref{f.K2_10}. As expected, for each $K$ both the fraction of nodes fully covered and the max tree depth converge, as in \Cref{f.K2}. 
In \Cref{f.k2_10number}, lines shift right almost linearly with a slow rate as $K$ increases. With $K=10$, over $90\%$ nodes are fully covered by time $40$. In \Cref{f.k2_10tree}, lines shift both downwards and right, and converge to lower values as $K$ increases. Because balanced trees have smaller depth if nodes have larger degree: with $K=10$, the line converges to $4$ in \Cref{f.k2_10tree}. Whatever $K$ is, the max tree depth is below $20$ by time $25$.

\begin{figure}
\begin{subfigure}[l]{0.22\textwidth}
	\post{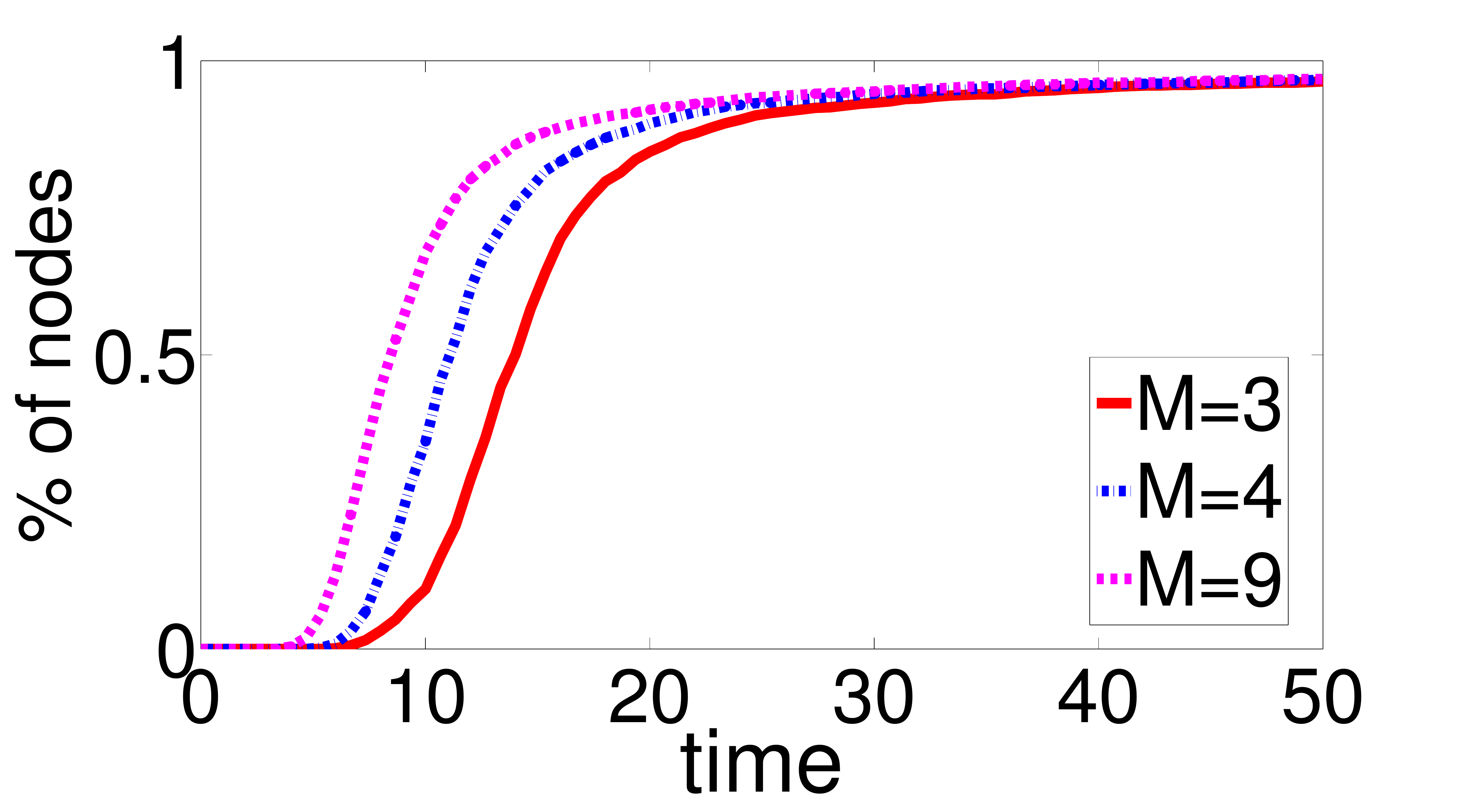}{4.3}
	\caption{$\%$ of nodes fully covered}
    \label{f.M4_7number}
\end{subfigure}
\begin{subfigure}[r]{0.22\textwidth}
	\post{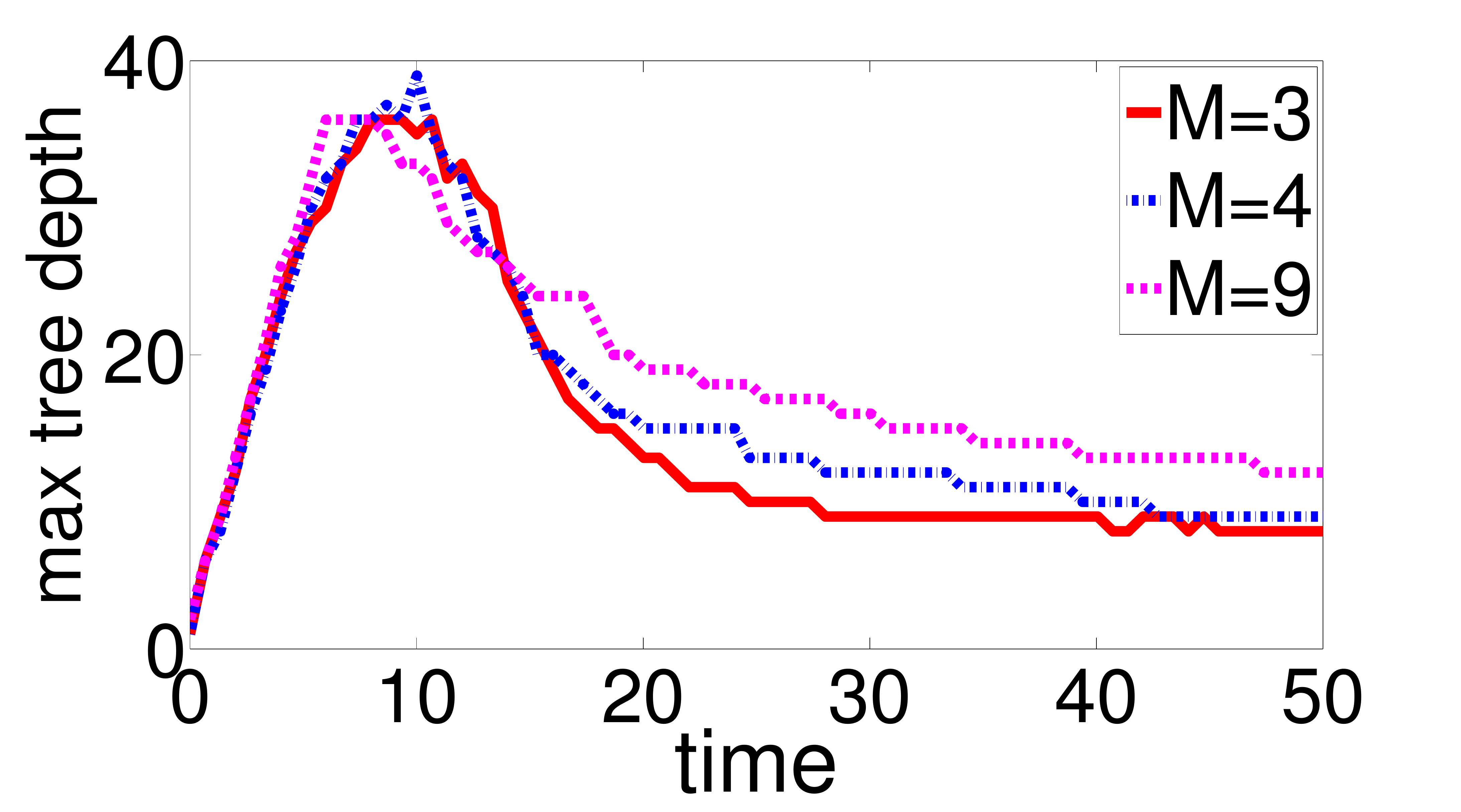}{4.3}
 	\caption{max tree depth}
\label{f.M4_7tree}
\end{subfigure}
\caption{$1\%$ cumulative lines when $M$ varies and $K=3$.}
\label{f.M4_7}
\end{figure}

In \Cref{f.M4_7}, we test the case under source coding by drawing the $1\%$ worst case lines when  $K=3$ and $M$ be in $\{3,4,9\}$, that is, there are more trees than nodes need. Notice that the capacity is still tight. In \Cref{f.M4_7number}, lines shift left as $M$ increases, showing that source coding tends to decrease the convergence time for coverage. In \Cref{f.M4_7tree}, limits slightly increase as $M$ increases. Because when $M$ is larger there are more types of mixed nodes, which may have single children in a tree and thereby increase the depth. That condition is also implied by \Cref{lem.conv}, where the value $c$ is considered to be of  $O(\log M)$. 
Thus, source coding creates a tradeoff between the tree depth limit and the convergence time of coverage. Intuitively and as shown in \Cref{f.M4_7tree}, the increasing of depth limit is of $O(\log M)$ which is small, so it is worth trying source coding to get a faster convergence.

For all simulations above, we test cases where the capacity is tight, which illustrates \Cref{lem.conv} and supports that convergence is exponential. One common feature of curves in \Cref{f.k2number,f.k2_10number,f.M4_7number} is that long tails exist. For example, in \Cref{f.M4_7number}, it takes quite long for curves to arrive at $1$. That is because near the end of the process only a few nodes are available and a few others are not fully covered, and it takes long for these nodes to meet each other by random sampling. Long tails can be eliminated by broadcasting or adding more capacity. Broadcasting is not discussed in this paper, in the following we show experiments where extra capacity exists.

\subsection{Loose capacity constraint eliminates the long tail}

\begin{figure}
\begin{subfigure}[l]{0.22\textwidth}
	\post{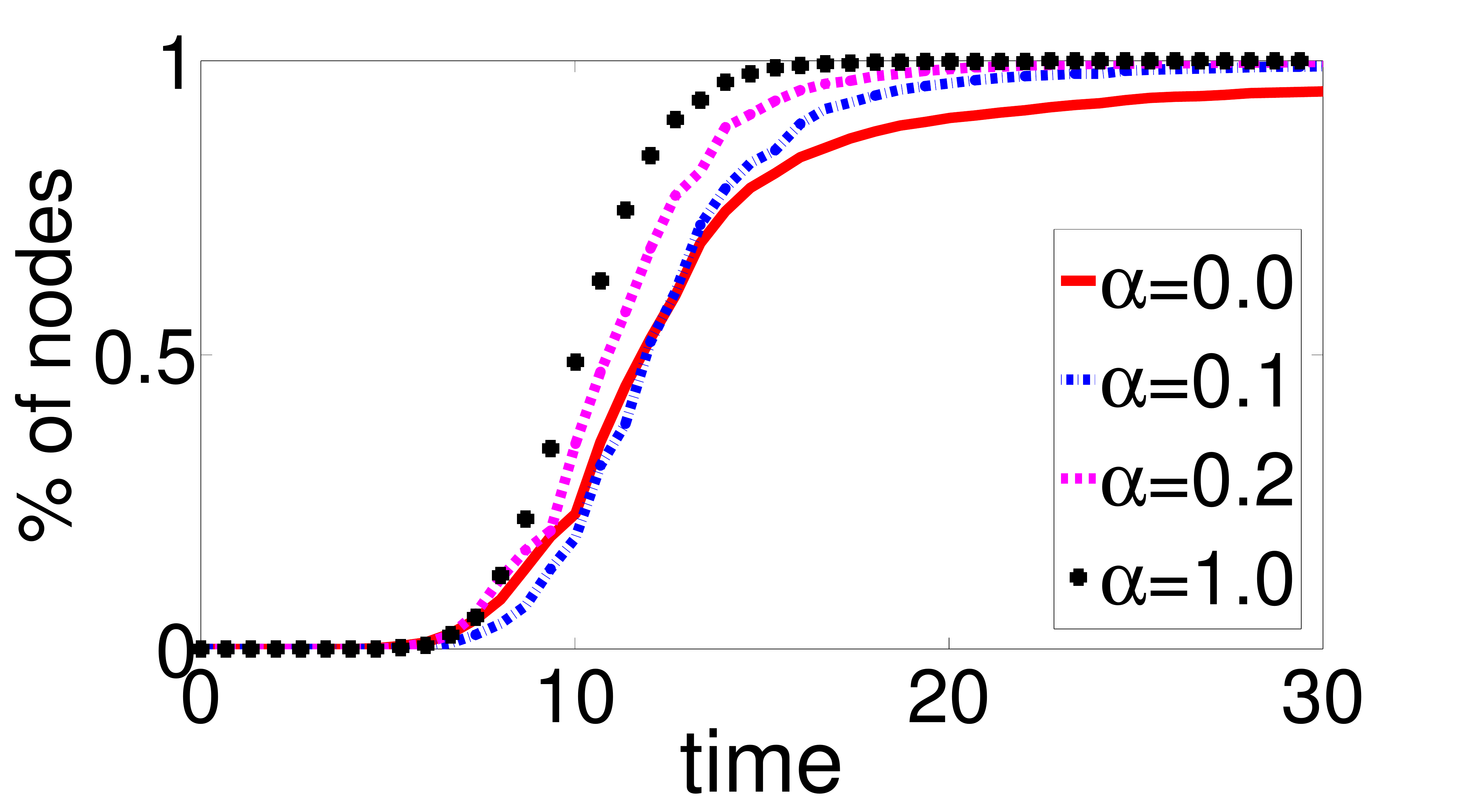}{4.3}
	\caption{$\%$ of nodes fully covered}
    \label{f.tight_rho_number}
\end{subfigure}
\begin{subfigure}[r]{0.22\textwidth}
	\post{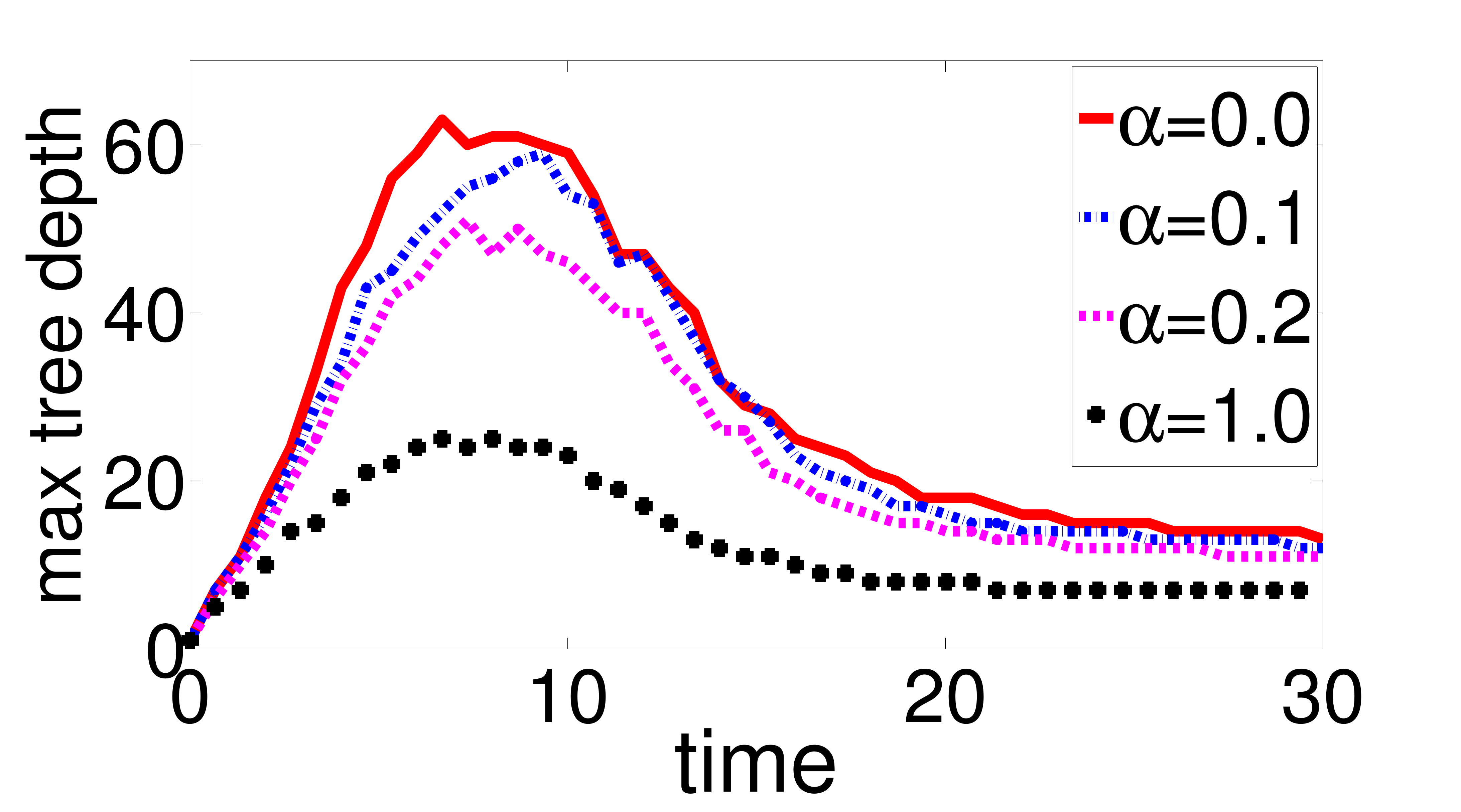}{4.3}
 	\caption{max tree depth}
\label{f.tight_rho_tree}
\end{subfigure}
\caption{$1\%$ cumulative, $M=K=2$ and $\alpha$ varies.}
\label{f.tight_rho}
\end{figure}

We set the total number of degrees $\sum_{u\in V}\bar{d}_u = (1+\alpha) NK$, with a new parameter $\alpha$. To achieve that, we first let each node have degree $K$, then add $\alpha NK$ degrees, one by one, to nodes selected uniformly at random. 
In \Cref{f.tight_rho}, we set $M=K=2$, and let $\alpha$ increase. In \Cref{f.tight_rho_number}, we can see that adding just $10\%$ extra capacity can greatly shorten the tail: all nodes are fully covered by time $25$ as shown by the line $\alpha = 0.1$. The larger $\alpha$ is, the shorter the tail is. When $\alpha=1.0$, all nodes get fully covered by time $15$. In \Cref{f.tight_rho_tree}, as $\alpha$ increases, curves converge faster and limits also decrease.

\subsection{Polarized degrees for the server-client case}

\begin{figure}
\begin{subfigure}[l]{0.22\textwidth}
	\post{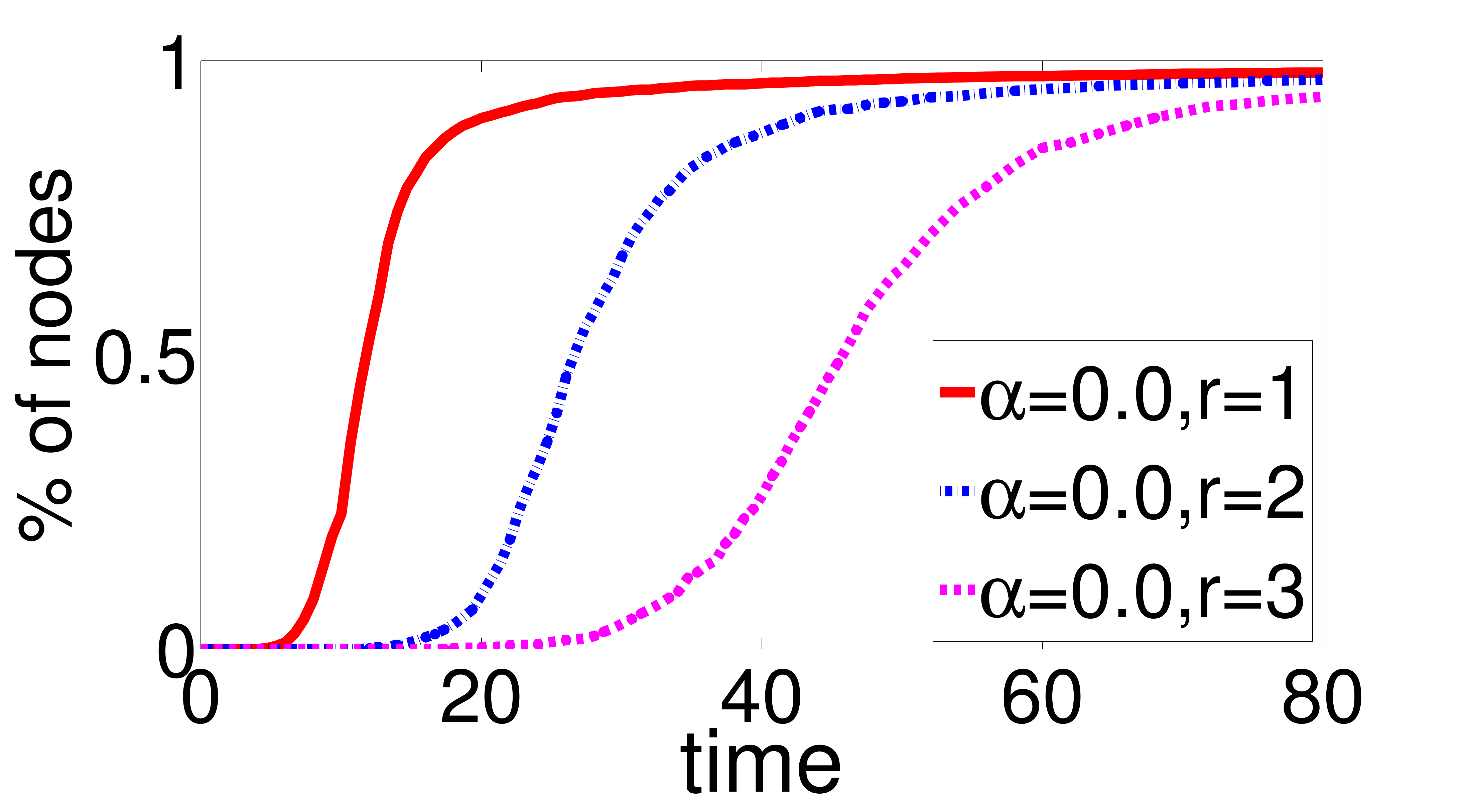}{4.3}
 	\caption{$\%$ of nodes fully covered}
	\label{f.K2_r_n}
\end{subfigure}
\begin{subfigure}[r]{0.22\textwidth}
	\post{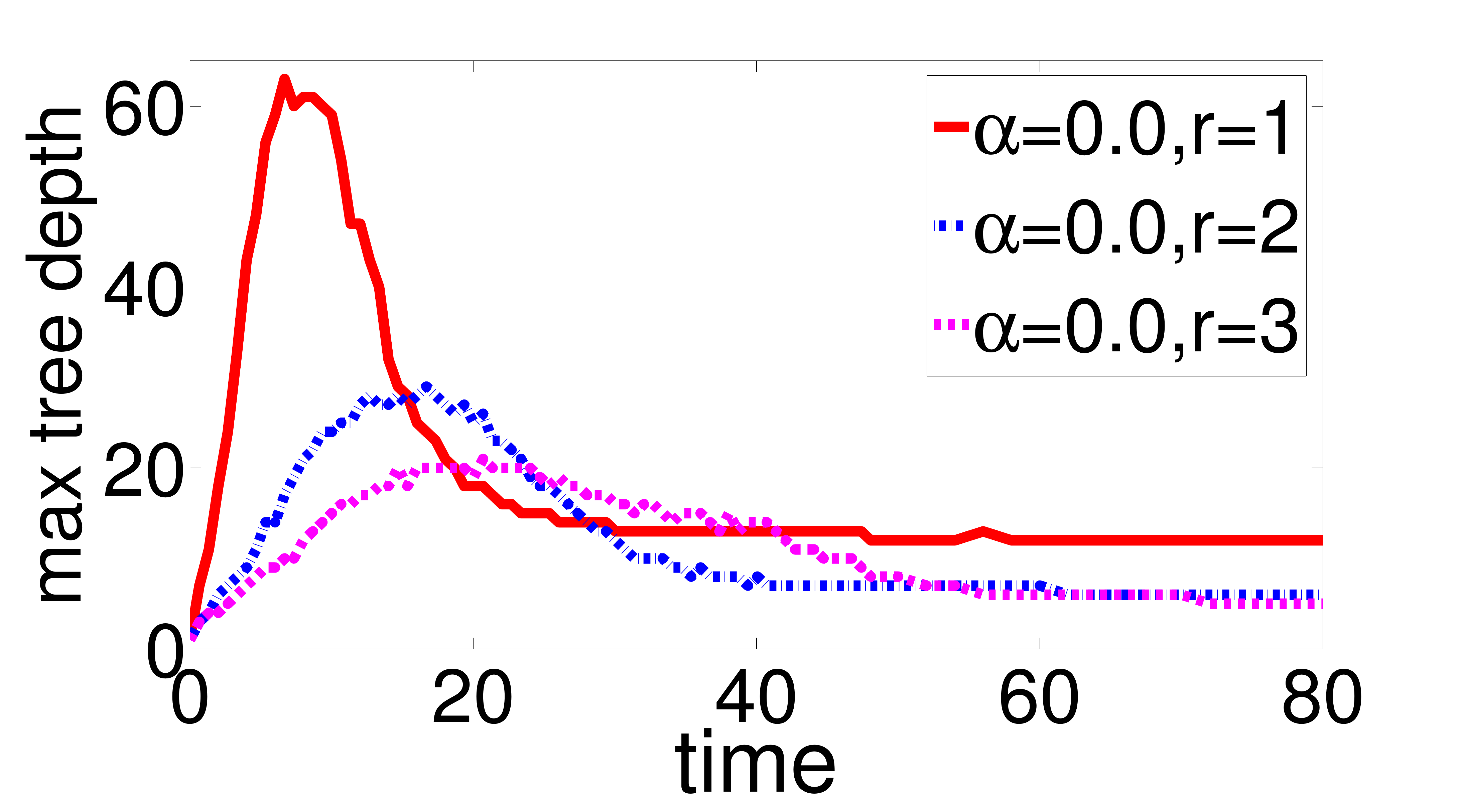}{4.3}
 	\caption{max tree depth}
\label{f.K2_r_t}
\end{subfigure}\\
\begin{subfigure}[l]{0.22\textwidth}
	\post{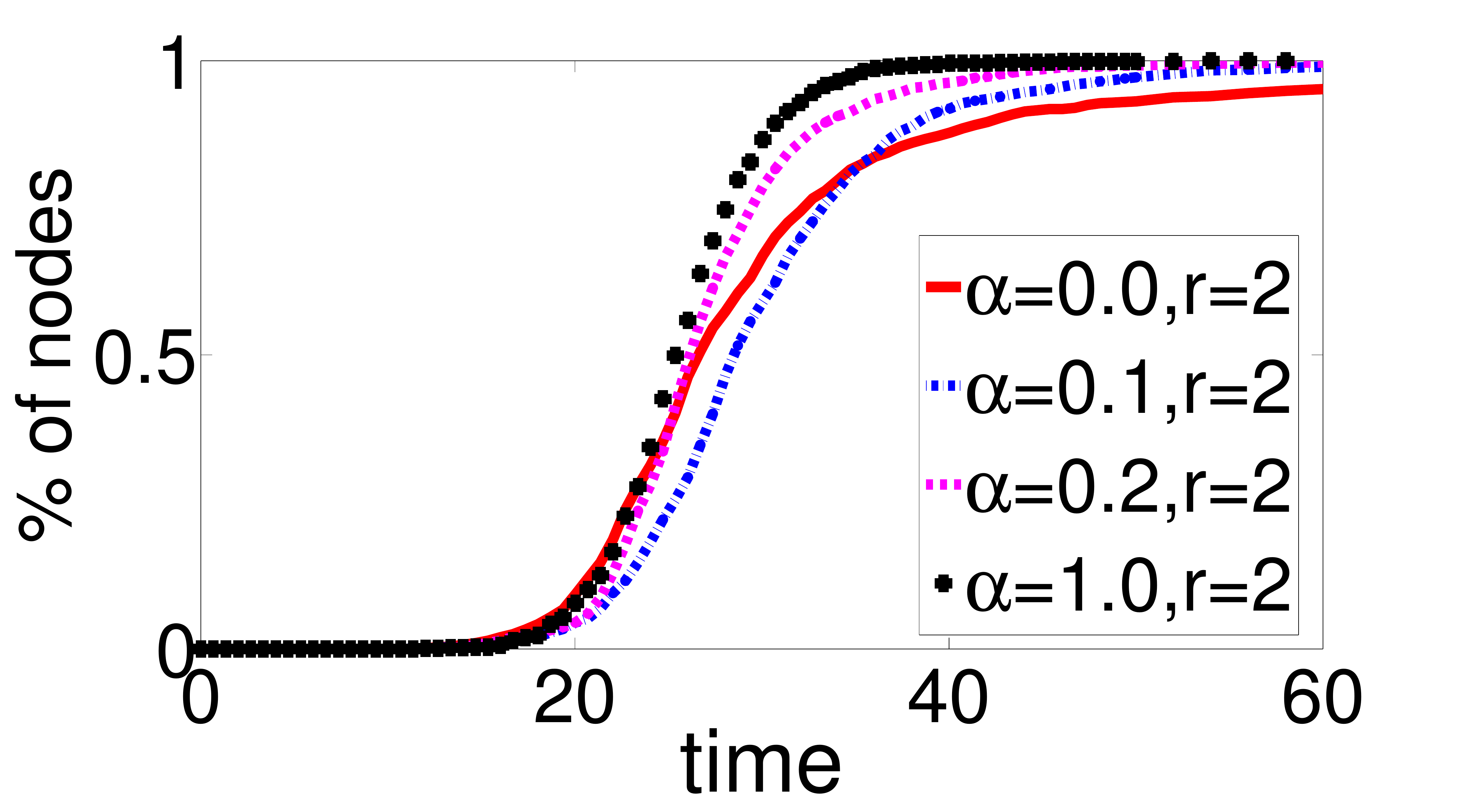}{4.3}
	\caption{$\%$ of nodes fully covered}
   \label{f.K2_degreeServer_n}
\end{subfigure}
\begin{subfigure}[r]{0.22\textwidth}
	\post{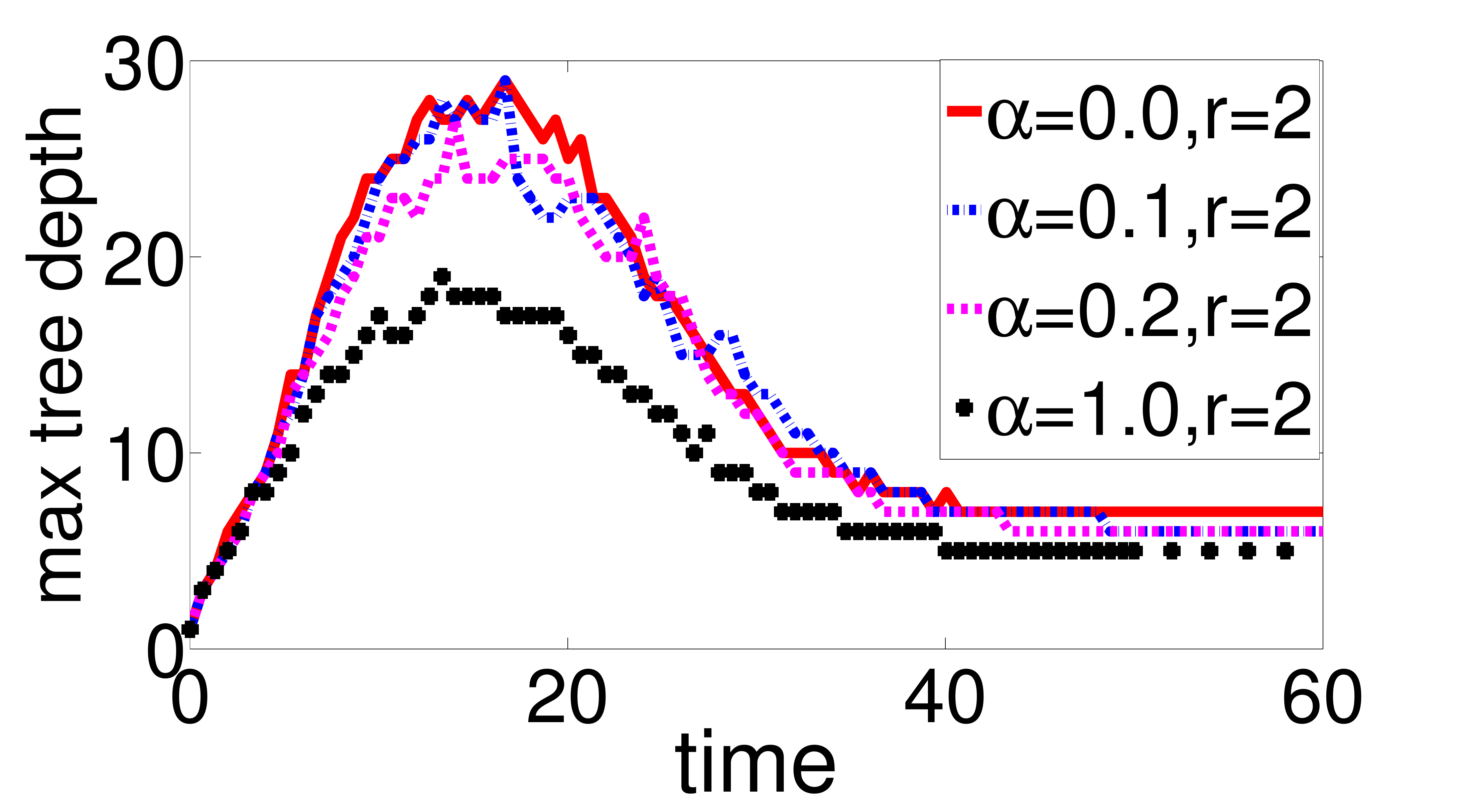}{4.3}
	\caption{max tree depth}
	\label{f.K2_degreeServer_t}
\end{subfigure}
\caption{$1\%$ cumulative. $M=K=2$, ${N/ r}$ server nodes with degree $rK$, $\alpha NK$ degrees added to server nodes randomly.}
\label{f.degreeToServer}
\end{figure}

\begin{figure}
\begin{subfigure}[l]{0.22\textwidth}
	\post{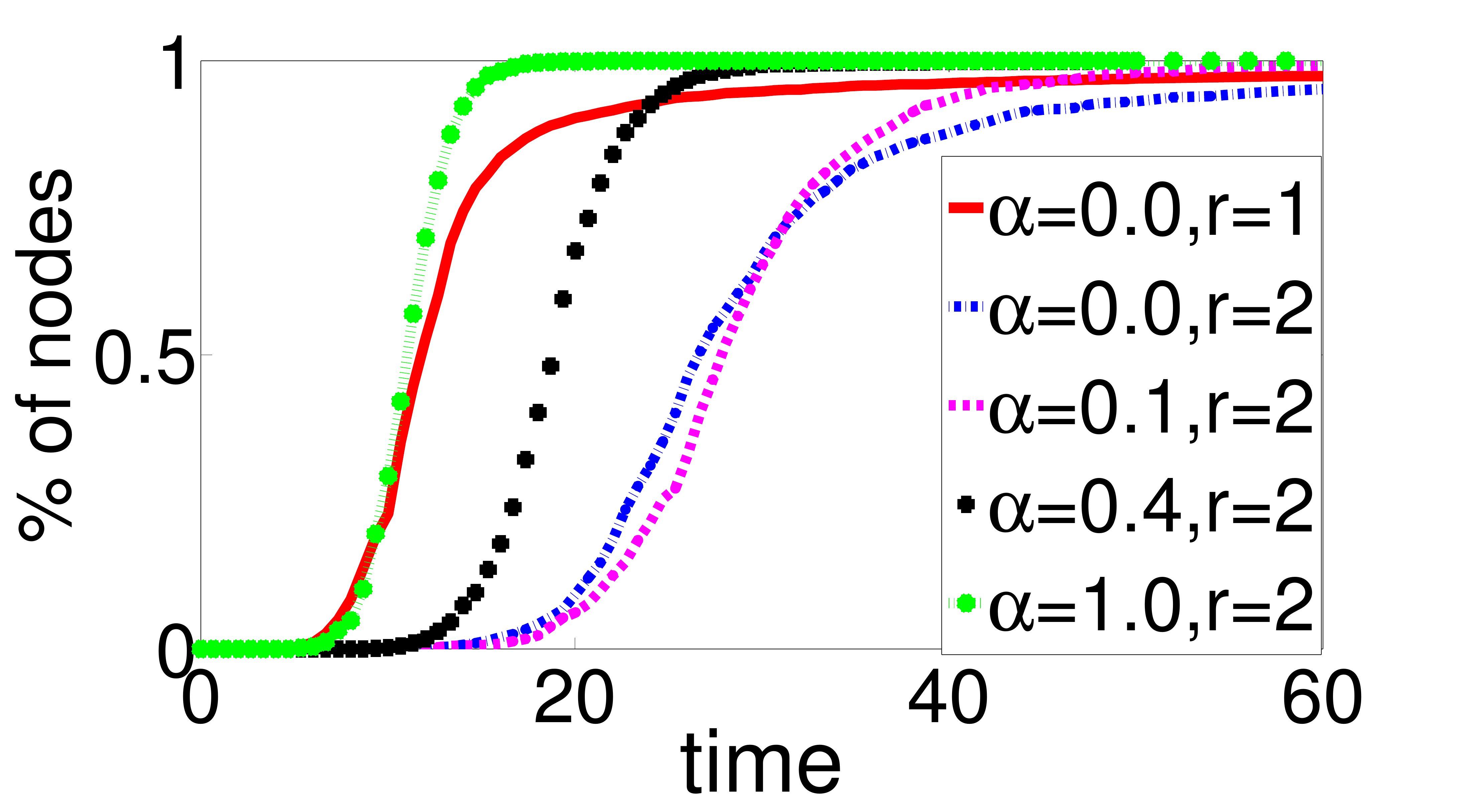}{4.3}
	\caption{$\%$ of nodes fully covered}
   \label{f.K2rho_n}
\end{subfigure}
\begin{subfigure}[r]{0.22\textwidth}
	\post{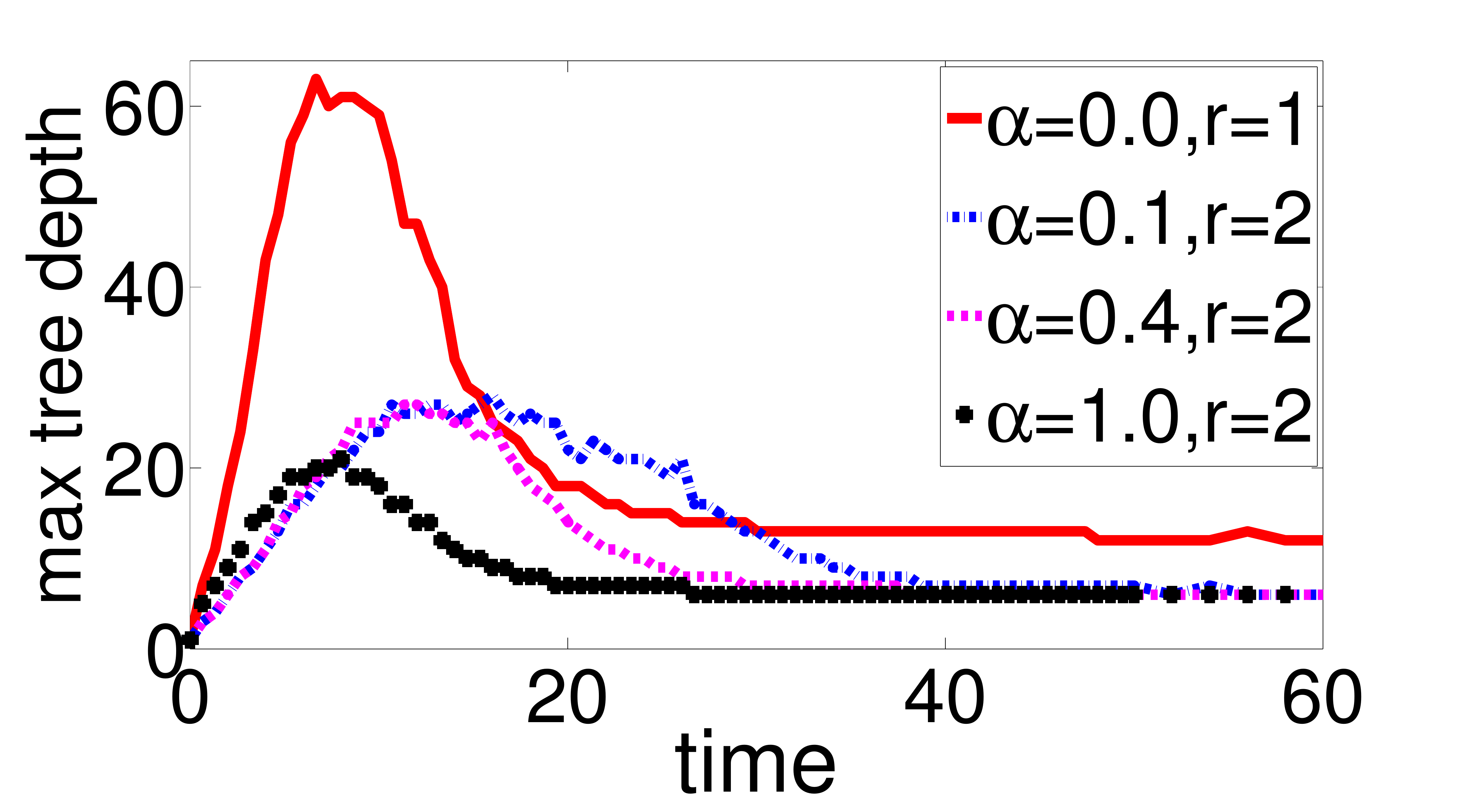}{4.3}
	\caption{max tree depth}
	\label{f.K2rho_t}
\end{subfigure}
\caption{$1\%$ cumulative. $M=K=2$, ${1+\alpha\over r} N$ server nodes with degree $rK$.}
\label{f.polarize}
\end{figure}

In this section, we show that \algwhole{0} works well under the server-client case, where a portion of nodes are server nodes with high degrees and other nodes are client nodes with zero degrees.  
The algorithm favors nodes with higher degree; they tend to get smaller depths than nodes with lower degrees. We expect to observe similar images as in \Cref{f.K2}. Notice that the convergence times will increase because it takes more time for nodes to meet server nodes under uniform random sampling. That problem can be solved by adding mechanisms to help nodes find server nodes. In this paper, we stay focused on the uniform sampling assumption despite the small increase in convergence time. 

In \Cref{f.degreeToServer}, we let $N/r$ server nodes (include roots) have degree $rK$ and all other nodes have degree $0$, then add $\alpha NK $ degrees one by one to server nodes randomly. As expected, in \Cref{f.K2_r_n}, lines shift right as $r$ increase because there are less server nodes, but still increases exponentially from $0$ to $1$. In \Cref{f.K2_r_t}, max tree depth decreases greatly as $r$ increase from $1$ to $2$, and further decreases as $r$ increase. In \Cref{f.K2_degreeServer_n,f.K2_degreeServer_t}, $r$ is fixed at $2$ with $\alpha$ increases, which shows long tails get eliminated like that in \Cref{f.tight_rho}.

In \Cref{f.polarize}, we let ${1+\alpha\over r} N$ server nodes (include roots) have degree $rK$ and all other nodes have degree $0$. We set $r=2$, let $\alpha$ increase, and draw the line at $\alpha=0,r=1$ for comparison.  Notice that lines at $\alpha=0$ in \Cref{f.polarize,f.degreeToServer} are exactly the same. 
 As $\alpha$ increases, there are more server nodes so lines shift left quickly in \Cref{f.K2rho_n}, and max tree depth decreases quickly in \Cref{f.K2rho_t}. When $\alpha=0.1,r=2$, lines in \Cref{f.polarize} are close to lines in \Cref{f.degreeToServer}. As $\alpha$ increases, in \Cref{f.polarize} lines shift left but in \Cref{f.degreeToServer} they do not. That suggests that performance is sensitive to the number of server nodes instead of the degree distribution among server nodes.

 Above all, our simulations validate \Cref{lem.conv,p.rate} by showing that fraction of nodes fully covered increases from $0$ to $1$, and the maximum tree depth decreases to its limit, almost exponentially, under tight or un-tight capacity constraint, homogeneous or heterogeneous capacity distribution,  when \Cref{a.init,a.root,a.link,a.2} hold. The simulations suggest that cycles are eliminated quickly because the fraction of nodes covered is almost increasing over all time; long tails can be eliminated by adding $10\%$ extra capacity; and the algorithm favors nodes with higher degree so it works well under the server-client case too. 
 Convergence times increase with either increasing $K$ or decreasing the chance for nodes to meet server nodes. When parameters change, curves shift with shapes staying similar, revealing  robustness of the algorithm.











\section{Conclusion}\label{sec_con}

In this paper a distributed algorithm to manage multiple
trees for streaming is proposed.  The algorithm can achieve
coverage, balance, and small delay in a short time. The
algorithm works on a complete underlying graph assuming
random sampling among peers is enabled. Future work
may include extending the algorithm for given incomplete
underlying topologies.

\section{Acknowledgments}
This work was supported by the National Science Foundation under Grant NSF CCF 10-16959.

\end{document}